%% file: lcp_submitted_version.tex
\documentclass[12pt, draftclsnofoot, onecolumn]{IEEEtran}	% regular single-column

\usepackage{fancyhdr}
\usepackage{epsfig}
\usepackage{threeparttable}
\usepackage{epsf,epsfig}
\usepackage{amsthm}
\usepackage[cmex10]{amsmath}
\usepackage{amssymb, amsfonts}
\usepackage[noadjust]{cite}
\usepackage{dsfont}
\usepackage{subfigure}
\usepackage{color}
\usepackage{soul}
\usepackage{amssymb}
\usepackage{accents}
\usepackage{algorithm}
\usepackage{diagbox}
\usepackage[english]{babel}
\usepackage{url}
\usepackage{framed}
\usepackage{multirow}
\usepackage{graphicx}
\usepackage{indentfirst}
\usepackage{url}
\usepackage{framed} 
\usepackage[colorlinks,
linkcolor=black,
anchorcolor=black,
citecolor=black]{hyperref}

\include{header}

\allowdisplaybreaks[2]
\def\W{\mathbf{W}}
\def\H{\mathbf{H}}
\def\U{\mathbf{U}}
\def\h{\mathbf{h}}
\def\V{\mathbf{V}}
\def\v{\mathbf{v}}
\def\Q{\mathbf{Q}}

\def\q{\mathbf{q}}
\def\t{\mathbf{t}}
\def\T{\mathbf{T}}
\def\bf#1{\mathbf{#1}}

\def\P{\mathscr{P}}

\usepackage{mathrsfs}
\usepackage{algorithm}
\usepackage{algorithmic}

\usepackage[justification=centering]{caption}
\begin{document}	
	\title{\huge A Deep Learning-Based Framework for Low Complexity Multi-User MIMO Precoding Design}
	\author{Maojun Zhang, Jiabao Gao, Caijun Zhong
	\thanks{\setlength{\baselineskip}{13pt} \noindent
	M. Zhang, J. Gao, and C. Zhong are with the College of Information Science and Electronic Engineering, Zhejiang University, Hangzhou, China (Email: $\{$zhmj, gao\_jiabao, caijunzhong$\}$@zju.edu.cn.)
	%, gao\_jiabao@zju.edu.cn, caijunzhong@zju.edu.cn).
	}	
	}
	% \author{\IEEEauthorblockN{Maojun Zhang\IEEEauthorrefmark{4}}
	% \IEEEauthorblockA{\IEEEauthorrefmark{4} College of information Science and Electronic Engineering, Zhejiang University, Hangzhou, China\\}}
	% 	\thanks{\setlength{\baselineskip}{13pt} \noindent G. Zhu is with the Shenzhen Research Institute of Big Data, Shenzhen, China (Email: gxzhu@sribd.cn). Y. Du and K. Huang are with the Dept. of Electrical and Electronic Engineering at The  University of  Hong Kong, Hong Kong (Email: yqdu@eee.hku.hk, huangkb@eee.hku.hk). Deniz G\"und\"uz is with the Dept. of Electrical and Electronics Engineering at the Imperial College of London, London, UK (Email: d.gunduz@imperial.ac.uk). Corresponding author: K. Huang. }
	% 	\thanks{\setlength{\baselineskip}{13pt} \noindent
	% 		The work of G. Zhu was supported in part by National Key R\&D Program of China under grant 2018YFB1800800, in part by National Natural Science Foundation of China under grant 62001310, in part by Guangdong Province Key Area R\&D Program under grant 2018B030338001, and in part by Shenzhen Peacock Plan under grant KQTD2015033114415450.
	% 		%and internal fund of Shenzhen Research Institute of Big Data with project ID J00120190020.
	% 		The work of K. Huang is supported by the Hong Kong Research Grants Council under the Grants 17208319 and 17209917.
	% 		The work of  D. G\"und\"uz received funding from the European Research Council (ERC) under Starting Grant BEACON (agreement no. 677854).}
	% }
	\maketitle
	%\vspace{-22mm}
	\begin{abstract}
		% \emph{Multiuser multiple-input multiple-output} (MU-MIMO) system has received great
		% attention in wireless communication since it can dramatically increase the
		% spectrum efficiency.
		% Precoding is necessary to reduce the interference between users.
		% While the conventional algorithm can't achieve a satisfied trade-off between complexity and performance.
		% In this work, with the help of powerful deep learning,
		% a precoding scheme with low complexity and high performance is proposed. 		
		% With a slightly performance loss the MIMO precoding problem is first converted to MISO precoding problem,
		% then with a optimal solution's structure, we resort to deep learning to tackling this non-convex problem.
		% Numeral experiments shows that the proposed scheme achieves a good trade-off between complexity and performance.
		Using precoding to suppress multi-user interference is a well-known technique to improve spectra efficiency in {multiuser multiple-input multiple-output} (MU-MIMO) systems, and the pursuit of high performance and low complexity precoding method has been the focus in the last decade. 
		The traditional algorithms including the zero-forcing (ZF) algorithm and the weighted minimum mean square error (WMMSE) algorithm failed to achieve a satisfactory trade-off between complexity and performance. 
		In this paper, leveraging on the power of deep learning, we propose a low-complexity precoding design framework for MU-MIMO systems. 
		{The key idea is to transform the MIMO precoding problem into the multiple-input single-output precoding problem, where the optimal precoding structure can be obtained in closed-form.} 
		A customized deep neural network is designed to fit the mapping from the channels to the precoding matrix. In addition, the technique of input dimensionality reduction, network pruning, and recovery module compression are used to further improve the computational efficiency. 
		Furthermore, the extension to the practical MIMO orthogonal frequency-division multiplexing (MIMO-OFDM) system is studied. Simulation results show that the proposed low-complexity precoding scheme achieves similar performance as the WMMSE algorithm with very low computational complexity. %of $\mathcal{O}\left(\rm {ZF}\right)$.
	\end{abstract}

	\newpage
	\section{Introduction}
	% {Multiple-input multiple-output} (MIMO) is a key technology for improving wireless communication capacity and has been evolved for generations,
	% among which, {multiuser multi-input multi-output} (MU-MIMO) has received great attention as it leverages spatial sparsity and thus enables MIMO system to serve multiple users simultaneously.
	% In the typical MU-MIMO system, precoding
	% is important and necessary to improve the spectrum efficiency.
	% By suppressing the  interference between different users or data streams, the sum rate of the MU-MIMO system can be dramatically improved.
	% Nowadays, the {fifth-generation} (5G) system's commercial process is advancing, and the {sixth-generation} (6G) system's research is also in progress.
	The multiple-input multiple-output (MIMO) antenna technique can dramatically improve the spectral efficiency by exploiting spatial domain resources, and has become a key enabling technology for the next generation wireless communication systems. Combining with multi-user spatial multiplexing, multi-user MIMO (MU-MIMO) can further increase the spectral efficiency and user capacity, which is crucial for supporting the massive access scenario in the future.
	To fully realize the potential of MU-MIMO, it is essential to preform precoding to suppress the interuser interference.
%	Due to the demands of ultra-reliable low latency communications (URLLC) in the fifth-generation (5G) systems and beyond,
%%	To meet the demands of ultra-low latency for real-time applications in 5G systems and beyond,
%	there are strict requirements for precoding, that is, the object of precoding design is not only better performance, but also less computational complexity.
%	This thus prompts an active research area focusing on developing novel computation-efficient precoding techniques.
	
	Many methods have appeared in the literature to tackle the precoding problem in MU-MIMO systems. Specifically, for the classic weighted sum rate (WSR) maximization problem, two different types of solution have been proposed, namely, optimization-based iterative algorithm and intuition-based non-iterative algorithm. The weighted minimum mean square error (WMMSE) algorithm is one of the most popular iterative algorithms \cite{christensen2008weighted,shi2011iteratively}. The original WMMSE algorithm is further improved in \cite{hu2020iterative,lu2020learning}. In general, the WMMSE algorithm can achieve near optimal sum rate performance. Yet, due to slow convergence rate and high-dimensional matrix inversion in each iteration, the WMMSE algorithm suffers from high computational complexity.
	%There has been some ways to improve the original WMMSE algorithm.
	%Realizing the optimization of the precoding matrix in WMMSE needs to find the optimal Lagrange multiplier,
	%authors in \cite{hu2020iterative} finds a equivalent problem, where the power constraint was removed.
	%Further the authors in \cite{lu2020learning} reduced the dimension of the matrix to be inverted in the iteration process by restricting the dimension of the solution space.
	%However, even with the above efforts, due to the iterative nature, the complexity of WMMSE algorithm is still unbearable in current communication system.
	Another feasible way is to solve the WSR problem heuristically,
	such as %the maximum-ratio transmission (MRT) algorithm and
	the zero-forcing (ZF) algorithm.
	%The former tries to maximize the signal power while the latter tries to minimize the interference between users.
	%Since they don't need to solve WSR problem by searching the optimal solution globally,
	%Due to the needlessness of  solving WSR problem by searching the optimal solution globally,
	Since the heuristic algorithms enjoy a closed-form solution and thus have much lower computational complexity.
	%Researchers focused on the ZF algorithm as it outperforms MRT algorithm in most scenarios.
	The performance of ZF-based algorithm under different system settings has been studied in \cite{caire2003achievable,spencer2004zero,yoo2006optimality,sun2010eigen,peel2005vector,lu2020block,zhang2018performance}.
	Since the ZF algorithm ignores the transmission noise and has a restrict solution space, its performance is unsatisfactory in the low {signal-to-noise-ratio} (SNR) regime or the fully loaded scenarios.
	In conclusion, the existing solutions suffer from either high computational complexity or unsatisfactory performance. 

%Therefore, it is highly desirable to there still needs further research to realize a truly computation-efficient precoding algorithm.
%	
%	As mentioned, the precoding algorithm with low complexity and high performance has not been obtained yet by improving the existing conventional algorithms.

	%and on the other hand,
	Recently, deep learning (DL) has demonstrated its superiority in the fields of computer vision and natural language processing, which motivates its application in the area of wireless communications. For instance, DL has been used to solve many classic communication problems, such as resource management \cite{sun2018learning,ye2019deep,liang2019spectrum,liang2019deep,eisen2019learning}, signal detection \cite{ye2017power,samuel2019learning,he2020model}, channel estimation \cite{he2018deep,dong2019deep,ma2020data,yang2020deep,gao2021attention} and beamforming \cite{gao2020unsupervised,song2020unsupervised,pellaco2020deep,hu2020iterative,huang2018unsupervised,xia2019deep,kim2020deep,lu2020learning}.
	The main advantage of DL is that, with rich training samples, the intensive computing task is transferred to the offline training stage,
	 only simple forward computation is needed in the online predicting stage. As such, DL provides a new way to design low complexity precoding algorithm. The DL-based design can be classified into two different types, namely model-driven design and data-driven design \cite{he2019model,zappone2019wireless}.
	The model-driven method tries to construct the neural network by leveraging on the prior expert knowledge.
	One of the most popular model-driven approaches is ``deep unfolding", which unfolds the existing iterative algorithm and embeds some trainable parameters to improve the performance.
	The unfolding design of WMMSE algorithm for precoding has been discussed in \cite{pellaco2020deep,hu2020iterative}. Specifically, the authors in \cite{pellaco2020deep} first proposed the idea of using deep unfolding of the WMMSE algorithm for MISO downlink precoding. Different variations of the deep unfolding idea have been proposed in \cite{shi2011iteratively,hu2020iterative}. Generally speaking, the deep unfolding schemes could achieve decent performance when compared with the original iterative algorithm. However, since the deep unfolding method retains the iterative process, the computational complexity is still high. On the other hand, the data-driven precoding algorithm has low computational complexity \cite{huang2018unsupervised}, albeit with significant performance degradation. 	%This has been verified that
	In most cases, it is more efficient to predict the key features by network and then recover the solution of the original problem through a recovery module.
	Aware of this, the authors in  \cite{hu2020iterative} and \cite{lu2020learning} proposed to take the intermediate iterative variables in WMMSE algorithm as the network output, and then recover the final precoding matrix through a recovery module.
	However, even with the recovery module designed in \cite{lu2020learning}, there is still noticeable performance gap with the WMMSE algorithm.
	%Unfortunately, it still can't achieve a comparable performance with WMMSE algorithm.
	
	To realize high-performance data-driven precoding algorithm, the key is to simplify the learning task of network, which can be achieved by fully exploiting the existing expert knowledge.
	Most recently, for the multi-input single-output (MISO) case, capitalizing on the optimal solution structure derived in \cite{yu2007transmitter,bjornson2014optimal},
	the authors in \cite{xia2019deep,kim2020deep} used convolutional neural network (CNN) to fit the mapping between two power allocation vectors and the channel matrix. Motivated by this, in the current paper, we first propose to transform the original MIMO precoding problem into a MISO precoding problem, and then design a low complexity DL-based scheme for the WSR maximization problem in MIMO systems. The extension to the practical MIMO orthogonal frequency-division multiplexing (MIMO-OFDM) system is also discussed.
	%Moreover, out of complexity concerns, we consider compressing the network and the recovery module.
	%Moreover, given that for the practical multiple-input multiple-output orthogonal frequency-division multiplexing (MIMO-OFDM) systems, the precoding granularity is usually more than one resource block (RB), we extend the proposed scheme to the scenario that multiple adjacent RB channels share a common precoding matrix, which is rarely considered in recent DL-based precoding literature.
	% The contributions of this work are summarized as follow.
	% Moreover, we improved the CNN based precoding scheme in \cite{xia2019deep} from input dimensionality reduction, network pruning and recovery module compression.
	% The proposed scheme achieves less computational complexity and better performance.
	% And considering that for practical MIMO systems, the precoding granularity is usually more than 1 resource block,  %(The channels within the resource block are considered the same.).
	% %we discuss extending the proposed scheme to practical precoding multiplexing problem,
	% we extend the proposed algorithm to the scenario that multi adjacent RB channels share a precoding matrix,
	% which is rarely considered in the recent deep learning precoding literature. %, and derived a corresponding optimal solution structure.
	% Simulation results shows that the proposed low complexity precoding (LCP) scheme significantly outperforms the ZF algorithm and the data driven scheme in \cite{hu2020iterative,lu2020learning}, and achieves near performance of WMMSE algorithm with similar computational cost of ZF.
	The contributions of the paper are summarized as follows:
	\begin{itemize}
		\item{We propose a DL-based low complexity downlink precoding framework for MU-MIMO systems. Capitalizing on the optimal precoder structure of the MISO systems, an approximated optimal precoder structure of the MIMO systems is presented.}
\item{The idea of input dimensionality reduction, network pruning and recovery module compression is proposed, which not only significantly improves the computational efficiency, but also provides a simple way to balance the system performance and computational complexity.}
		\item {
		We extend the proposed scheme to the practical MIMO-OFDM scenario where multiple adjacent resource blocks (RB) share a common precoding matrix. Based on the newly derived optimal solution structure, a dedicated neural network is proposed to learn the mapping from the channel to the key features.}
		%Realizing that the precoding granularity is usually required multi RB in practical MIMO system,
		%we extend the proposed scheme to a practical precoding multiplexing problem. First  a corresponding optimal solution structure is derived,
		%Then a specified neural network is designed to learn the mapping from the channel to three key features.}
		\item {Simulation results show that the proposed scheme can achieve similar performance as the WMMSE algorithm, but with much lower computational complexity. In addition, it was demonstrated that the proposed scheme has strong generalization ability, and can adapt to various scenarios with different number of user streams, varying number of users, or imperfect channel state information (CSI). }
	\end{itemize}

	The rest of this paper is organized as follows. Section II introduces the system model and problem formulation.
	Section III gives the method of transforming the MIMO precoding problem into MISO precoding problem. Based on the optimal solution structure, the novel learning based low complexity precoding design is presented in Section IV.
	Then, Section V discusses the extension to MIMO-OFDM systems. Extensive simulation results are presented in Section VI.
	Finally, Section VII concludes the paper.

	\emph{Notations:} We use italic, bold-face lower-case and bold-face upper-case letter to denote scalar, vector and matrix, respectively.
	$\mathbf{A}^T$ and $\mathbf{A}^H$ denote the transpose and Hermitian or complex conjugate transpose of matrix $\mathbf{A}$, respectively.
	$\mathbf{A}^{-1}$ denotes the inversion of matrix $\mathbf{A}$.
	$\mathbb{E}\{\cdot\}$ denotes the statistical expectations.
	${\rm Tr}\{\cdot\}$ denotes the trace operation.
	$\|\mathbf{x}\|$ denotes the $l$-2 norm of vector $\mathbf{x}$. $\mathbb{C}^{x\times y}$ denotes the $x\times y$ complex space.
	$\mathcal{CN}(\mu,\sigma^2)$ denotes the distribution of a circularly symmetric complex Gaussian random variable with mean $\mu$ and variance $\sigma^2$.
	$\mathcal{U}\left[a,b\right]$ denotes the uniform distribution between $a$ and $b$.
%	while for the SRM problem in MIMO system, to the best of our knowledge,
%	there is no such a simple solution structure

	% The deep learning proposed by the Geoffrey Hinton team in 2006 pushed the research and applications of Artificial Intelligence (AI) to an unprecedented height.
	% With years of development, deep learning has been proved its superiority in the fields of speech recognition, image recognition, and drug activity prediction, where the traditional technologies have reached bottlenecks.
	
	% \textbf{deep learning based precoding algorithm and notes the existing problem}
	% For the AI based system, the generalization ability is bad, and the goodness is it can convert the computational complexity to offline training.
	% The main

	% \textbf{illustrate the proposed precoding scheme}

	% \textbf{contribution:}

	% \textbf{}
	% first discuss the MU MIMO system, and then note that the importance of precoding, then finally tell the reader that a precoding algorithm that simultaneously low complexity and high performance is very important for the current and future wireless communication system.
	% the model-driven based problem has one main drawback, that is it has bad system generalization ability, like user's data stream changing and user's num changing.
	% review the previous conventional algorithm, including the heuristic algorithm, e.g., ZF, and iterative algorithm,
	% then we should say that the iterative algorithm could achieve nearly optimal performance but suffer from high computational complexity

	\section{System Model and Problem Formulation}
	\subsection{System Model}\label{subsec:system model}
		We consider a single cell downlink MU-MIMO system,
	where the BS is equipped with $N_t$ transmit antennas and serves $K$ users each with $N_r$ antennas.
	The received signal at the $k$-th user $\mathbf{y}_k\in \mathbb{C}^{N_r\times 1}$ is given by
	\begin{align}\label{eq:trceived signal}
		\mathbf{y}_k =\mathbf{H}_k\mathbf{V}_k \mathbf{s}_k+\sum_{m=1,m\neq k}^{K}\mathbf{H}_k\mathbf{V}_m \mathbf{s}_m+\mathbf{n}_k, \forall k \in \mathcal{K}
	\end{align}	
	where $\H_k \in \mathbb{C}^{N_r \times N_t}$ denotes the channel matrix from the BS to user $k$,
	$\mathbf{V}_k\in\mathbb{C}^{N_t\times d_k}$ denotes the precoding matrix for user $k$,
	$d_k$ is the number of data streams,
	$\mathbf{n}_k \in \mathbb{C}^{N_r \times 1}$ denotes the additive white Gaussian noise with distribution $\mathcal{CN} (0,\sigma^2 \mathbf{I})$,
	$\mathbf{s}_k\in\mathbb{C}^{d_k\times 1}$ denotes the transmitted data for user $k$ that satisfies $\mathbb{E}[\mathbf{s}_k \mathbf{s}_k^H]=\mathbf{I}$, and it is assumed that the transmitted data between different users is independent, $\mathcal{K}=\left\{1,2,...,K\right\}$ denotes the set of users.

	% The transmitted signal is given as follow
	% \begin{align}\label{eq:transmitted signal}
	% 	x=\sum_{k=1}^{K}\mathbf{W}_k s_k,
	% \end{align}
	% where $\mathbf{W}_k\in\mathbb{C}^{N_t\times d}$ denotes the precoding matrix of user $k$,
	% d is the number of data stream, $s_k\in\mathbb{C}^{d\times 1}$ denotes the transmitted data for user $k$.
	% $s_k$ satisfies $\mathbb{E}[s_k s_k^H]=\mathbf{I}$, and it is assumed that the transmitted data between different users is independent.
	% The received signal $y_k \in \mathbb{C}^{N_r\times 1}$ can be expressed below:

	% where $\H_k \in \mathbb{C}^{N_r \times N_t}$ denotes the channel matrix from the BS to user $k$,
	% and $n_k \in \mathbb{C}^{N_r \times 1}$ is the additive white Gaussian noise with distribution $\mathcal{CN} (0,\sigma^2 \mathbf{I})$.

	% Here we consider Gaussian multi-path channel, $\H_k$ is given by
	% \begin{align}\label{eq:chanenl model}
	% 	\H_k = \sum_{i=1}^{L_p} \alpha_i \bf{a}_r (\theta_i^r) \bf{a}_t (\theta_i^t),
	% \end{align}	
	% where $\alpha_i \sim \cn(0,\sigma_{\alpha}^2)$ denotes the path gain, $L_p$ is the number of path,
	% $\bf{a}_r(\theta_{i}^{r}) \in \mathbb{C}^{N_r \times 1}$, $\bf{a}_t(\theta_{i}^{t}) \in \mathbb{C}^{1 \times N_t}$ denotes the antenna array factor.
	\subsection{Problem Formulation}\label{subsec:problem formulation}

	For the downlink precoding design, we consider to maximize the weighted sum rate (WSR) under total transmit power constraint, the problem is given below
	\begin{align}\label{pr:sum rate max P1}
		\mathscr{P}_1: \max_{\left\{\V_k\right\}}&\sum_{k=1}^{K}\alpha_k R_k,\\
		\mathbf{s.t.}\ &\sum_{k=1}^{K} {\rm Tr} (\mathbf{V}_k\mathbf{V}_k^H)\leq P_T, \notag
	\end{align}	
	where $P_T$ denotes the transmit power budget of the BS, $\alpha_k$ is a weighted scalar which represents the priority of user $k$.
	%User $k$'s achievable rate $R_k$ is given by
	The achievable rate of user $k$, $R_k$ is given by
	\begin{align}\label{eq: R def in P1}
		R_k = \log \det\left(\bf{I} + \H_k \V_k \V_k^H \H_k^H \left(\sum_{m\neq k}\H_k \V_m \V_m^H \H_k^H +\sigma^2 \bf{I}\right)^{-1} \right).
	\end{align}

	It is obvious that $\mathscr{P}_1$ is nonconvex, and thus hard to solve.
	In prior works, an iterative WMMSE algorithm has been proposed to find the optimal solution.
	Yet, it incurs high computational cost and substantial delays, which makes it difficult to implement in practice.
	In the rest of the paper, leveraging on the power of DL technique,
%	The iterative WMMSE algorithm is a feasible solution to achieve high performance, while due to the high computational cost, it usually introduces an unbearable delay.
	%In the rest of the paper, with the powerful deep learning,
	we propose a
	DL-based method which achieves comparable performance as the WMMSE algorithm with much lower complexity.
	%method to achieve similar performance with WMMSE algorithm but have much lower complexity.
	%The existing precoding algorithm including ZF and WMMSE are either low complexity and low performance, or high complexity and high performance.
	%An ideal algorithm with high performance and low complexity has not been proposed yet.
	
	%here the problem need to be illustrated.
	\vspace{-2mm}
	\section{Problem Transformation}
	In this section,
	we start with the special MISO case, and present the optimal structure of the downlink precoding matrix.
	Then by transforming the MIMO precoding problem into the MISO precoding problem, the optimal structure in the general MIMO scenario is designed.
%	we first consider a special case of $\mathscr{P}_1$, and obtain the optimal structure of the precoding matrix.
%	Then, to utilize the optimal structure in the general MIMO scenario, we develop a method to transform the MIMO precoding problem into the MISO precoding problem.
%	we first discuss the $\mathscr{P}_1$ in the MISO ($N_r =1$) scenario, and the optimal solution structure is obtained.
%	But there is difficulty extending the optimal solution structure to the MIMO scenario,
%	so we resort to transforming the original MIMO precoding problem into the MISO precoding problem.
	%But there exists difficult of extending the optimal solution structure to MIMO scenario,
	%But this optimal solution structure can not be simply extended to MIMO scenario,
	%we thus resort to transforming the original MIMO precoding problem to the MISO precoding problem.
	\vspace{-3mm}
	\subsection{Optimal Solution Structure}
	We first introduce the principles of the well known WMMSE algorithm \cite{hu2020iterative,shi2011iteratively}.
	It repeats the following three steps until convergence:
	\begin{align} \label{eq:wmmse iterative update}
		\mathbf{U}_k&=\left(\frac{\sigma^2}{P_T}\sum_{k=1}^K {\rm Tr}\left(\mathbf{V}_k\mathbf{V}_k^H\right)\mathbf{I}_{N_r}+\sum_{m=1}^K\mathbf{H}_k\mathbf{V}_m\mathbf{V}_m^H\mathbf{H}_k^H\right)^{-1}\mathbf{H}_k\mathbf{V}_k,~~\forall k,\\
		\mathbf{W}_k&=\left(\mathbf{I}_{d_k}-\mathbf{U}_k^H\mathbf{H}_k\mathbf{V}_k\right)^{-1},~~\forall k,\label{eq:wmmse w update}\\
		\mathbf{V}_k&=\alpha_k \left[\sum_{m=1}^K \left(\frac{\sigma^2}{P_T}{\rm Tr}\left(\alpha_m \mathbf{U}_m\mathbf{W}_m\mathbf{U}_m^H\right)\mathbf{I}_{N_t}+\alpha_m\mathbf{H}_m^H\mathbf{U}_m\mathbf{W}_m\mathbf{U}_m^H\mathbf{H}_m\right)\right]^{-1}\mathbf{H}_k^H\mathbf{U}_k\mathbf{W}_k,~~\forall k. \label{eq:wmmse v update}
	\end{align}

	Therefore, the optimal precoding matrix $\V_k^*$ should meet the form of (\ref{eq:wmmse v update}).
	%When considering
	For the MISO system (i.e., $N_r=d_k=1, \forall k$), the intermediate iteration variables $\mathbf{U}_k$ and $\mathbf{W}_k$ reduce to  scalars.
	To avoid confusion, we use $\left\{\h_k, u_k, w_k, \v_k\right\}$ to denote $\left\{\H_k, \U_k, \W_k, \V_k\right\}$ in the MISO scenario.
	Then the optimal precoding vector in (\ref{eq:wmmse v update}) can be rewritten as
	\begin{align}\label{eq:optimal solution's structrue without normalization}
		\mathbf{v}_k^*=\gamma_k\left(\sigma^2\mathbf{I}+\sum_{m=1}^K\lambda_k\mathbf{h}_m^H\mathbf{h}_m\right)^{-1}\mathbf{h}_k^H,
	\end{align}
	%where $\lambda_k=P_T\frac{u_mw_mu_m^H}{\sum_{n=1}^Mu_nw_nu_n^H},
	where $\lambda_k=P_T\frac{\alpha_m u_mw_mu_m^H}{\sum_{n=1}^K\alpha_n u_nw_nu_n^H}, \gamma_k=P_T\frac{\alpha_m u_m w_m}{\sum_{n=1}^K \alpha_n u_nw_nu_n^H}$ are scaling factors.
	It can be observed that $\v_k$ and $e^{j\theta_k}\v_{k}$ have the same objective value in $\mathscr{P}_1$ \cite{bjornson2014optimal},
	thus we have the following proposition.
	\begin{proposition}\label{prop:single RB miso precoding optimal solution structure}
		\emph{For the MU-MISO system with channel $\left\{\h_k, k=1,...,K\right\}$, the optimal solution of $\P_1$ can be expressed as}
	\begin{align}\label{eq:miso optimal sulition structrue}
		\v_k^* = \sqrt{p_k} f\left[\left(\sigma^2\mathbf{I}+\sum_{m=1}^K\lambda_m\mathbf{h}_m^H\mathbf{h}_m\right)^{-1}\mathbf{h}_k^H\right],
	\end{align}
	\emph{where $\lambda_k$ and $p_k$ are positive scaling factors satisfying $\sum_{k=1}^K\lambda_k=\sum_{k=1}^K p_k=P_T$.
	$f[\mathbf{x}]=\mathbf{x}/\left\|\mathbf{x}\right\|_2$ is the normalization function.}
	\end{proposition}
	\begin{remark}
		\emph{
		For the WSR problem in the MISO system, the optimal solution depends on two factors, %'s structure can be divided into two parts,
		namely the item $\left\{\left(\sigma^2\mathbf{I}+\sum_{m=1}^K\lambda_m\mathbf{h}_m^H\mathbf{h}_m\right)^{-1}\mathbf{h}_k^H\right\}$ determining the beamforming direction, and the item $\left\{\sqrt{p_k}\right\}$ determining the beamforming power.
		Moreover, to obtain the optimal precoding vector $\v^*$,
		only $\mathbf{p}$ and $\boldsymbol{\lambda}$ need to be calculated,
		where $\mathbf{p}$ is the downlink power allocation vector, and $\boldsymbol{\lambda}$ can be viewed as the virtual uplink power allocation vector \cite{bjornson2014optimal}.
		}
	\end{remark}
	With the optimal solution in hand, the original WSR problem $\P_1$ can be converted to two power allocation problems.
	The powerful neural network can be used to learn the mapping from the channel to $\boldsymbol{\lambda}=[\lambda_1,...,\lambda_K]^T$ and $\mathbf{p}=[p_1,...,p_K]^T$,
	as discussed in \cite{xia2019deep}.

	However, to the authors' best knowledge, there is no such optimal solution structure of $\P_1$ for MIMO systems (i.e., $N_r >1$).
	To utilize the optimal solution structure in (\ref{eq:miso optimal sulition structrue}) for MIMO systems, we resort to transforming the MIMO precoding problem into the MISO precoding problem.
	\subsection{Transform MIMO Precoding to MISO Precoding}\label{subsec: miso 2 miso}
	Without loss of generality, we assume $N_r < N_t$.
	%In most practical MIMO systems, $N_r\ll N_t$.
	Performing the truncated singular value decomposition (SVD) for $\H_k$, we have
	\begin{align}\label{eq:channel svd decomposition}
		\H_k=\Q_k\boldsymbol{\Sigma}_k\T_k^H,
	\end{align}
	where $\Q_k\in \mathbb{C}^{N_r\times N_r}$, $\T\in \mathbb{C}^{N_t\times N_r}$, $\boldsymbol{\Sigma}_k=\diag\left(\sigma_{1,k},...,\sigma_{N_r,k}\right)$ denote the left singular matrix, right singular matrix and square matrix composed of singular values, respectively.

	As $\Q_k$ can be regarded as a part of recovery precoding at the receiver, for downlink precoding design, $\boldsymbol{\Sigma}_k\T_k^H$ can be interpreted as the equivalent channel matrix. 
    {Thus the equivalent received signal $\widetilde{\mathbf{y}}_k = \mathbf{Q}_k^{-1}\mathbf{y}_k$ is given as follows. 
	%is equivalent to precoding for $\H_k$.
	\begin{align}
		\widetilde{\mathbf{y}}_k = \left[\sigma_{1,k}\mathbf{t}_{1,k},...,\sigma_{N_r,k}\mathbf{t}_{N_r,k}\right]^H\sum_{m=1}^M\left[\mathbf{v}_{1,m},...,\mathbf{v}_{d_m,m}\right]\mathbf{s}_m + \widetilde{\mathbf{n}}_k,
	\end{align}
	where $\sigma_{i,k}$ denotes the $i$-th largest singular value, with $\t_{i,k}\in \mathbb{C}^{N_t \times 1}$ being the corresponding right singular vector, $\mathbf{v}_{i,k}$ being the precoding vector for data $s_{i,k}$, $s_{i,k}$ being the $i$-th stream of $\mathbf{s}_m$.}

	{Further separating the received signal of the $i$-th antenna of user $k$, we have  
	\begin{align}\label{eq: representatino of y_i_k}
		\widetilde{y}_{i,k} &= \sigma_i\mathbf{t}_{i,k}^H\sum_{m=1}^M \sum_{j=1}^{d_m}\mathbf{v}_{j,m}s_{j,m} + \widetilde{n}_{i,k}\notag\\&=\underbrace{\widehat{\h}_{i,k}\mathbf{v}_{i,k}s_{i,k}}_{\rm desired~signal} + \underbrace{\widehat{\h}_{i,k}\sum_{m=1,m\neq k}^M \sum_{j=1}^{d_m}\mathbf{v}_{j,m}s_{i,m}}_{\rm interference~from~other~ user's ~streams}+\underbrace{\widehat{\h}_{i,k}\sum_{l=1, l\neq i}^{d_k}\mathbf{v}_{l,m}s_{l,k}}_{\rm interference~from~same~user's~streams}+\widetilde{n}_{i,k},
	\end{align}}
	{It can be seen that the transmission process of each data stream can be viewed as a MISO system with the channel $\widehat{\h}_{i,k} = \sigma_{i,k}\t_{i,k}^H$, where $\widehat{\h}_{i,k}$ can be considered as the virtual channel. 
	As shown in (\ref{eq: representatino of y_i_k}), 
	we propose to manually assign a virtual channel to each data stream, other streams from both other users and the same user are treated as  interference. }
	%Moreover, other streams from both other users and the same user are deemed interference.  
	%For MISO downlink transmission, $N_r=d_k=1$, i.e., each data stream occupies a specified physical channel.
	%For MIMO downlink transmission,
	%the number of user $k$'s data streams $d_k$ is usually less than the number of receiving antennas $N_r$.
	%The $d_k$ data streams are jointly transmitted through $N_r$ physical channels, making the problem complicated.
	%Therefore,
	%we propose to assign a virtual channel to each data stream.}
	From the perspective of SVD for the channel, there are $N_r$ mutually orthogonal channels with different channel gains,
	among which, we assign the channels with the highest gains to the $d_k$ data streams {\footnote{{Notice that, the proposed transformation method is only applicable when $N_r \geq d_k, \forall k$. In case of $N_r < d_k$, proper scheduling algorithm can be applied in practice.}}},
	%We assign the most strong channel to the $d_k$ data streams,
	which is given below.
	\begin{align}\label{eq:channel assign}
		\widehat{\H}_k = \left[\sigma_{1,k}\t_{1,k},...,\sigma_{d_k,k}\t_{d_k,k}\right]^H,
	\end{align}
	where $\sigma_{i,k}$ denotes the $i$-th largest singular value, with $\t_{i,k}\in \mathbb{C}^{N_t \times 1}$ being the corresponding right singular vector. %of $\sigma_{i,k}$.
	The virtual channel $\sigma_{i,k}\t_{i,k}$ is used for user $k$'s $i$-th data stream transmission.
	Now define $\left\{\widehat{\H}_k\right\}$ as follows.
	\begin{align}\label{eq:miso channel merge definition}
		\widehat{\H} = \left[\widehat{\H}_1^H,...,\widehat{\H}_K^H\right]=\left[\widehat{\h}_1^H,...,\widehat{\h}_M^H\right],
	\end{align}
	where $M = \sum_{k=1}^K d_k$ denotes the number of total data streams,
	$\widehat{\h}_m\in\mathbb{C}^{1\times N_t}$ can be viewed as the channel between the BS and a user with single receiving antenna, and has a priority weight coefficient inherited from the original MIMO problem  $\beta_m = \alpha_n$, where $\sum_{k=1}^{n-1}d_k<m\leq\sum_{k=1}^n d_k$.
	To this end, the original MIMO WSR precoding problem becomes equivalent to
	%To solve the original MIMO WSR precoding problem, we turn to solve
	the $M$ users MISO WSR problem with the channel $\left\{\widehat{\h}_m, m\in \mathcal{M}\right\}$, where $\mathcal{M}=\left\{1,2,...,M\right\}$ denotes the equivalent user set.
	The detailed process is summarized in Algorithm 1.
	It is worth noting that the transforming process incurs a performance loss by discarding some available channels.
	As demonstrated in Section \ref{subsec: problem transformation loss}, the performance loss is in general negligible.
	%while as the most strong channels are guaranteed to be used,
	%from the experiments in Section \ref{subsec: problem transformation loss}, we can see that the performance loss is negligible in most scenarios.
	\begin{algorithm}
		\caption{MIMO precoding transformation}\label{algo:sampling with relpacement}
		\textbf{Input:} Original channel matrix $\left\{\mathbf{H}_k\right\}$. \\
		\textbf{Output:} Precoding matrix $\left\{\mathbf{V}_k\right\}$. \\
		%\textbf{Initialize:} $\widehat{\mathbf{H}}=[~]$.\\
		1: \textbf{for:} $k=1:1:K$ \textbf{do}\\
		2: ~~Do truncated SVD for $\mathbf{H}_k$: $\H_k=\Q_k\boldsymbol{\Sigma}_k\T_k^H$, where $\boldsymbol{\Sigma}_k={\rm diag}\left(\sigma_{1,k},...,\sigma_{N_r,k}\right)$, $\mathbf{T}_k=\left[\mathbf{t}_{1,k},...,\mathbf{t}_{N_r,k}\right]$.\\
		3: ~~$\widehat{\H}_k = \left[\sigma_{1,k}\t_{1,k},...,\sigma_{d_k,k}\t_{d_k,k}\right]^H$.  \\
		%4: ~~$\widehat{\mathbf{H}}=\left[\right]$
		4: \textbf{end for}\\
		5: $\widehat{\mathbf{H}}=\left[\widehat{\H}_1^H,...,\widehat{\H}_K^H\right]=\left[\widehat{\h}_1^H,...,\widehat{\h}_M^H\right]$. \\
		6: Obtain the precoding matrix $\left\{\widehat{\mathbf{v}}_m\right\}$ by solving the MISO WSR problem, where there are $M=\sum_{k=1}^K d_k$ users in total, with the channel $\left\{\widehat{\mathbf{h}}_m\right\}$ and priority weight $\left\{\beta_m=\alpha_n \right\}$, satisfying $\sum_{k=1}^{n-1}d_k<m\leq\sum_{k=1}^n d_k$.  \\ %, and there are $M=\sum_{k=1}^K dk$ equivalent users in total. \\
		%7: $\mathbf{V}_1=\left[\mathbf{v}_1,...,\mathbf{v}_{d_1}\right]$\\
		%7: $i = 0$\\
		7: \textbf{for:} $k=1:1:K$ \textbf{do}\\
		%8: ~~$j=i + d_k$\\
%		8: ~~\textbf{if:} $k=1$ \\
%		9: ~~~~$i=1$,  \\
%		10: ~~\textbf{else:} $i=\sum_{n=1}^{k-1}d_k$. \\
%		9: ~~$j=\sum_{n=1}^k d_k$. \\
		8: ~~$\mathbf{V}_k=\left[\widehat{\mathbf{v}}_{(k-1)d_k +1},\widehat{\mathbf{v}}_{(k-1)d_k +2},...,\widehat{\mathbf{v}}_{kd_k}\right]$.\\
		%9: ~~$i=i+d_k$\\
		9: \textbf{end for}
	\end{algorithm}
	%By solving the $M$ users MISO SRM problem,
	%	The optimal solution structure has already discussed in [cite here], while the derivation is built on the relativeness of $\mathscr{P}_1$ and the
	%As deep learning is a powerful way to fit the reflection between
    \section{Learning Based Low Complexity Precoding Design}\label{sec:single RB deep learning based precoding algorithm}
	In this section, we will give the DL-based solution for the step 6 in Algorithm \ref{algo:sampling with relpacement}.
	We first discuss the design of neural network.
	Then to further reduce complexity, %out of complexity concerns,
	we consider the compression of recovery module and network pruning.
	Finally, the generalization ability of the proposed scheme and computational complexity are discussed.
	\begin{figure}[tt]
		\centering
		\includegraphics[width=14cm]{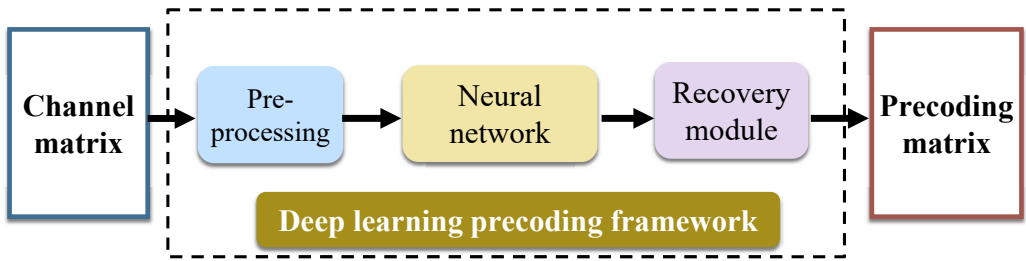}
		\caption{A general DL-based solution for MIMO precoding}
		\vspace{-10mm}
		\label{Fig:network structure for MIMO precoding}
	\end{figure}
	\subsection{Neural Network Design}
	To realize high-performance precoding algorithm,
	instead of using the ``black-box'' approach, it is essential to exploit expert knowledge, which can be achieved through pre-processing and post-processing.
	In general, an efficient deep learning based precoding method should consist of the following three key steps as illustrated in Fig \ref{Fig:network structure for MIMO precoding}:
	1) input pre-processing, which aims to remove the redundant information and reduce the input dimension;
	2) feature extraction, which intends to extract the most significant features from the compressed input;
	3) efficient recovery, which enables high quality recovery based on the key features leveraging the expert knowledge.
	The specific implementation of each sub module will be discussed below.
	%it is necessary to use the existing expert knowledge to carry out some pre-processing and post-processing work, rather than all in ``blackbox".
	%Though the neural network has been proved to have strong fitting ability, if the input is redundant or the output form is complex, the fitting effect of the neural network will be affected, or a larger network is required to achieve the same performance.
	%Therefore,
	%An efficient deep learning based precoding algorithm should be divided into the following three steps as shown in Fig \ref{Fig:network structure for MIMO precoding}: Firstly, preprocess the input. On the one hand, remove the redundant information. On the other hand, reduce the computational complexity by reducing the input dimension;
	%Then the neural network needs to extract the key features from the compressed input; Finally, a recovery module based on the expert knowledge is responsible for recovering the required precoding matrix from the key features.
	%The specific implementation of each sub module will be discussed below.
	\subsubsection{Pre-processing and Post-processing}\label{subsubsec: network input and output design}
	The objective of pre-processing and post-processing is to simplify the input and output of the network as much as possible.
	To simplify the output, we need to leverage the optimal solution structure in (\ref{eq:miso optimal sulition structrue}), which is achieved by transforming the original MIMO precoding problem to MISO precoding problem according to Algorithm 1.
	In this way, the network only needs to output two power allocation vectors $\mathbf{p}$ and $\boldsymbol{\lambda}$ instead of the whole precoding matrix, then  (\ref{eq:miso optimal sulition structrue}) can be used as the recovery module to recover the precoding matrix from $\mathbf{p}$ and $\boldsymbol{\lambda}$.
	Moreover, for such a feature learning task, the existing works \cite{xia2019deep,kim2020deep} directly use the channel matrix as the network input, while it is actually redundant for predicting $\mathbf{p}$ and $\boldsymbol{\lambda}$.
	To compress the network input, here we first give the following proposition.

	% First to leverage the optimal solution structure in (\ref{eq:miso optimal sulition structrue}), the MIMO channel is first converted to MISO channel according to Algorithm 1.
	% %Then to solve the corresponding MIMO WSR precoding problem, we resort to the powerful neural network.
	% For the corresponding MISO WSR precoding problem,
	% most of the existing works directly take the channel matrix as the input of the network \cite{xia2019deep,huang2019fast,kim2020deep},
	% while it is actually redundant for the key feature learning, and may affect the network's learning performance.
	% In order to streamline the input while without affecting performance, here we first give a proposition.
	\begin{proposition}\label{proposition:input instruct}
		\emph{For WSR maximization problem in MISO system with channels $\left\{\h_{k},\forall k\right\}$ and weights $\left\{\alpha_k,\forall k\right\}$ of $M$ users,
		%Let $\widehat{\mathbf{H}}\triangleq \left[\widehat{\h}_1^H,...,\widehat{\h}_M^H\right]$,
		the beamforming matrix's key features $\left\{p_k,\lambda_k,\forall k\right\}$ in (\ref{eq:miso optimal sulition structrue}) is determined by $\left\{\alpha_k,\mathbf{h}_i\mathbf{h}_j^H,\forall i, j ,k\right\}$.} %$\widehat{\H}\widehat{\H}^H. $%$\left\{\mathbf{h}_i\mathbf{h}_j^H\right\}$.
	\end{proposition}
	\begin{proof}
		See Appendix \ref{subsec: proof of proposition 2}.
	\end{proof}
	%According to priority weight processing in \cite{lu2020learning},
	Let $\overline{\mathbf{H}}\triangleq \left[\sqrt{\beta_1}\widehat{\h}_1^H,...,\sqrt{\beta_M}\widehat{\h}_M^H\right]$, where $\widehat{\h}_k$ is the virtual channel.
	With Proposition \ref{proposition:input instruct} in hand, we just need the inner product of each equivalent user channel vector as the network input rather than the full channel state information.
	Therefore, $\mathbf{\overline{\mathbf{H}}}^H\mathbf{\overline{\mathbf{H}}}$ is taken as  the input of the neural network,
	compared with the full channel matrix $\overline{\H}$, the input dimension has been reduced from $\left(N_t\times M\right)$ to $\left(M \times M\right)$.
	{On the other hand, as the neural network only supports real number operation, it is necessary to convert the input to the real number field. 
    The conventional method is to split the real and imaginary parts and then stack them together. 
    However, since $\overline{\mathbf{H}}^H\overline{\mathbf{H}}$ is a Hermitian matrix, the upper triangular part or lower triangular part of the matrix contains all the input information. 
    Therefore we adopt a more compact way proposed in \cite{lu2020learning}, that is, we first take the upper triangle of $\overline{\mathbf{H}}^H\overline{\mathbf{H}}$, then reassemble the real part and imaginary part into a square matrix.}
    %Moreover, as the neural network only supports real number, we use the method in \cite{lu2020learning} to triangulate and recombine the real and imaginary parts of $\mathbf{\overline{\mathbf{H}}}^H\mathbf{\overline{\mathbf{H}}}$.
	
	%The corresponding output is $\left\{p_k\right\}$ and $\left(\lambda_k\right)$.
	%with $\left\{p_k,\lambda_k\right\}$ the output.
%	The key feature $\left\{p_k,\lambda_k\right\}$ is the network's output.
	\subsubsection{Network Backbone Design}
	With pre-processing and post-processing,
	we adopt the network architecture in \cite{xia2019deep} to fit the mapping between $\overline{\mathbf{H}}^H\overline{\mathbf{H}}$ and $\left\{p_k,\lambda_k,\forall k\right\}$.
	%further the neural network should be designed to fit the mapping between $\widehat{\mathbf{H}}\widehat{\mathbf{H}}^H$ and $\left\{p_k,\lambda_k,\forall k\right\}$.
	%As the input is a square matrix, followed by the network architecture in \cite{xia2019deep}.
    {As shown in Fig. 2, 
	 the input of the network is the correlation matrix of the equivalent channels, hence, 
the powerful convolutional neural network is a suitable choice for feature extraction \cite{chen2019artificial}.  
The combination of the 2D convolutional layer, the batch normalization layer and the leaky relu activation layer composes the backbone of the network. 
Then the fully-connected layer combined with the sigmoid activation layer are responsible for predicting $\mathbf{p}$ and $\boldsymbol{\lambda}$ from the extracted information.\footnote{{In \cite{kim2020deep}, the authors proposed to estimate $\boldsymbol{\lambda}$ by $\mathbf{p}$ to further reduce the complexity. 
Since the $\mathbf{p}$ and $\boldsymbol{\lambda}$ share a common network backbone, this method can not efficiently reduce the complexity and will cause a performance loss, thus is not considered in this paper.}}}
	Considering that $\mathbf{p}$ and $\boldsymbol{\lambda}$ have constant $l$-1 norm, the outputs are further normalized by the respective $l$-1 norm.
	The network is trained by combining supervised and unsupervised learning.
	In the supervised learning stage,
	the label is obtained using the WMMSE algorithm, and the model is trained with the MSE loss function.
	Then, unsupervised learning for WSR maximization is used to further improve the performance.
	%the model is trained with the MSE loss function and the feature's label obtained by the MISO version WMMSE algorithm.
%	And in the unsupervised learning stage, the opposite number of weighted sum rate is used as loss function to further improve the performance.
	% the 2D convolutional layer is used for feature extracting, then fully connected layer is used to get the output from the key feature.
	% Moreover, as $\mathbf{p}$ and $\boldsymbol{\lambda}$ have a constant L1 norm, the outputs are further normalized by their L1 norm.
	\begin{figure}[tt]
		\centering
		\includegraphics[width=14cm]{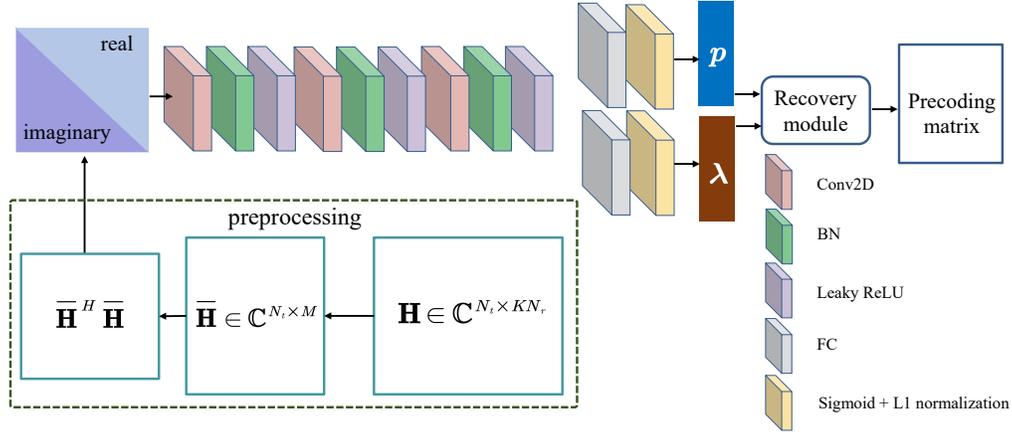}
		\caption{Network structure for MIMO precoding}
		\vspace{-7mm}
		\label{Fig:network structure single RB}
	\end{figure}

	%\subsubsection{Network Output Design}
	\subsection{Model Compression}
	\subsubsection{Recovery Module Compression}
	To obtain the final precoding matrix, the recovery module is needed to recover the precoding matrix from $\mathbf{p}$ and $\boldsymbol{\lambda}$, which conducts the process as in (\ref{eq:miso optimal sulition structrue}).
	Unfortunately,	 there exists a high-dimensional matrix inverse operation in (\ref{eq:miso optimal sulition structrue}), which introduces high computational cost.
	So here we consider a low complexity recovery method.
	First let $\mathbf{V}=\left[\mathbf{v}_1,...,\mathbf{v}_M\right]$, the recovery process can be rewritten as follows.
	\begin{align}\label{eq:miso precoding matrix first presentation}
		\widehat{\mathbf{V}} = F\left(\left(\sigma^2\mathbf{I}_{N_t} + \widehat{\mathbf{H}}\boldsymbol{\Lambda}\widehat{\mathbf{H}}^H\right)^{-1}\widehat{\mathbf{H}}\right)\mathbf{P}^{\frac{1}{2}},
	\end{align}
	where $\widehat{\mathbf{H}}$ is defined in (\ref{eq:miso channel merge definition}), $\boldsymbol{\Lambda}={\rm diag}\left(\lambda_1,...,\lambda_M\right)$, $\mathbf{P}={\rm diag}\left(p_1,...,p_M\right)$, and $F(\left[\mathbf{a}_1,...,\mathbf{a}_N\right])=\left[\frac{\mathbf{a}_1}{\left\|\mathbf{a}_1\right\|^2},...,\frac{\mathbf{a}_N}{\left\|\mathbf{a}_N\right\|^2}\right]$ is the column-wise normalization function.

	Let $\widetilde{\mathbf{H}}=\widehat{\mathbf{H}}\boldsymbol{\Lambda}^{\frac{1}{2}}$,
	then (\ref{eq:miso precoding matrix first presentation}) can be rewritten as $\widehat{\mathbf{V}} = F\left(\left(\sigma^2\mathbf{I}_{N_t} + \widetilde{\mathbf{H}}\widetilde{\mathbf{H}}^H\right)^{-1}\widetilde{\mathbf{H}}\boldsymbol{\Lambda}^{-\frac{1}{2}}\right)\mathbf{P}^{\frac{1}{2}}$.
	With the equivalent transformation $\left(a\mathbf{I}+\mathbf{A}\mathbf{A}^H\right)^{-1}\mathbf{A}=\mathbf{A}\left(a\mathbf{I} + \mathbf{A}^H\mathbf{A}\right)^{-1}$ \cite{nguyen2019multi}, where $a>0$,
	finally we have
	\begin{align}\label{eq:miso precoding matrix final presentation}
		\widehat{\mathbf{V}} = F\left(\widetilde{\mathbf{H}}\left(\sigma^2\mathbf{I}_{M} + \widetilde{\mathbf{H}}^H\widetilde{\mathbf{H}}\right)^{-1}\boldsymbol{\Lambda}^{-\frac{1}{2}}\right)\mathbf{P}^{\frac{1}{2}}.
	\end{align}
	With the equivalent transformation from (\ref{eq:miso precoding matrix first presentation}) to (\ref{eq:miso precoding matrix final presentation}),
	the dimension of the inverse matrix is reduced from $N_t$ to $M$, in most of the practical MIMO systems, $M<N_t$,
	therefore, the computational complexity of the restore module can be reduced.

	\subsubsection{Network Pruning}
	Artificially designed neural network generally have redundant parts, so it is necessary to do model compression.
%	Model compression aims to reduce the size of models while minimizing the performance loss.
	For the CNN model in Fig \ref{Fig:network structure single RB}, 
		%since the computational complexity of the convolutional layers occupies the main part of the whole network, 
		network structured pruning is an efficient way to remove the unnecessary filters of the convolutional layers. 
	We adopt the pruning algorithm in \cite{li2016pruning}, which removes the filters as well as  their connecting parts.	
	We apply it to remove the redundant filters of the first two convolutional layers.
	The whole pruning process can be divided into three steps:
	\begin{itemize}
		\item{{\textbf{Large Model Training:}}} In this stage, the complexity concern is omitted.
		Each convolutional layer is assigned enough filters to guarantee the best convergence performance.
		\item{{\textbf{Model Pruning:}}} {With the well-trained large model, for each convolutional layer, the filters are sorted based on their $l$-2 norm, and the filter has the smallest $l$-2 norm is viewed as unimportant \cite{li2016pruning,he2018soft}, thus needs to be pruned.}
		\item{{\textbf{Fine-Tuning:}}} Perform fine-tuning training on the pruned model with a small learning rate until convergence.
	\end{itemize}
	%With a slightly performance loss, the computation complexity can be significantly reduced.
	By repeating the last two steps above, the redundant parts can be gradually removed.
	%Moreover, there is another advantage.
	%By network pruning,
	In addition, it is worth noting that with network pruning,
	the deep learning based precoding algorithm can achieve dynamic adjustment between complexity and performance.
	In practice, different systems may have different requirements for the computational complexity.
	By pruning more filters, the complexity of network can be further reduced, at the cost of performance loss.
	Compared with the conventional iteration based algorithms (e.g., WMMSE algorithm), which reduce the computational complexity by adjusting the number of iterations,
	network pruning has a finer adjustment range, and thus is more practical.
	%Here further discuss the goodness of doing network pruning for the CNN model,
	%first with the pruning tools, the redundant parts can be removed, and moreover, generally speaking, a bigger model usually has a better fitting ability, so with network pruning, for different complexity requirements, the deep neural network based solution can achieve the dynamic adjustment between complexity and performance.
	\subsection{Generalization}
	The main advantage of deep learning method is that it offloads most of the high complexity operations to the training stage.
	%By saving the model parameters, it only needs to use the well-trained model for forward calculation in the application stage.
	Only simple forward calculation is required in the application stage.
	However, the major issue is that the model is trained for a specific system configuration.
	%The problem is that a specific model can only be applied to a specific system configuration.
	In the current wireless communication system, the system configuration may often change dynamically, such as number of users, number of data streams and imperfect CSI caused by the channel estimation.
	In this subsection, we will discuss the generalization performance of the proposed method in the above scenarios.
	%Existing works can only handle part of them, here we will

	%In the practical MIMO system,
	\subsubsection{Varying number of user streams} \label{subsec: vary ds}
	In the practical system, due to different data transmission requirements, the number of user streams may change. %, but by doing scheduling, the total number of system's steams is usually unchanged.
	This problem has not been properly addressed in existing deep learning based methods.
%	The existing deep learning based method can't deal with this case,
	For the proposed algorithm, since the original MIMO precoding problem is transformed into the MISO precoding problem,
	the neural network only needs to predict the key feature of the MISO precoding problem with $M=\sum_{k=1}^K d_k$ user streams.
	Therefore, as long as the total number of data streams remains unchanged, the trained network is applicable for any MIMO system with $\left\{d_k\right\}$, where $d_k$ satisfies $1\leq d_k \leq d_{max}$ and $\sum_{k=1}^K d_k = M$. %$\left\{d_k, 1\leq d_k \leq d_{max}, \forall k, \sum_{k=1}^K d_k = M\right\}$.
	\subsubsection{Varying number of users and antennas}
	The change of the number of users ($K$) and antennas ($N_r$, $N_t$) is also very common.
	We assume that we have obtained the well-trained model for a group of specific system parameters.
	For the system with the  same antennas but less user,
	the zero filling method can be applied.
	%we can fill zeros for the channel of redundant users.
	For the system with different numbers of receive or transmit antennas but serve same number of users, the proposed precoding scheme can be directly applied without any additional processing.
	This is because with the pre-processing  in Section \ref{subsubsec: network input and output design}, the input dimension is decreased to ($M \times M$), the information in antenna dimension has been well extracted in the pre-processing stage.
	%The generalization of $K$ and $N_t , N_r$ has also been discussed in \cite{hu2020iterative}.
	%With the trained model from system ($K_0, N_{r_0}, N_{t_0}$),
	%for their model driven method, their model can be applied to systems with parameter ($K_1 \leq K_0, N_{r_1}\leq N_{r_0}, N_{t_1}\leq N_{r_0}$), while for the data-driven method we have proposed, the model can be applied to systems with parameters ($K_2 \leq K_0, \forall N_{r_2}, \forall N_{t_2}$).
	%It shows that the  proposed scheme has a better generalization ability in number of users and antennas.
	%\subsubsection{Varying number of antennas at the transmitter and receiver}
	\subsubsection{Imperfect CSI}
	Many existing works on MIMO precoding assume that the CSI is perfect.
	However, in the practical system, the CSI is obtained through channel estimation, and thus may have errors (i.e., $\widetilde{\H}_k = \H_k + \mathbf{n}$).
	The advantage of data-driven method is that it can adapt to such an error through training.
	%That is,
	When preparing channel data for offline training, we assume that both the noisy channel and noiseless channel can be obtained,
	the former is obtained through a practical channel estimator, while the latter is obtained through an ideal channel estimator that does not account for channel estimation overhead.
	%in the offline training stage, we assume that a more accurate channel information can be obtained.
	The noisy channel $\widetilde{\H}_k$ is used for forward calculation of network and precoding recovery, and the noiseless channel $\H_k$ is used to calculate the final sum-rate loss function for unsupervised learning.
	In this way, the relationship between $\widetilde{\H}_k$ and $\H_k$ can be obtained by the network itself through back propagation,
	so that in the online prediction stage, to obtain the desired precoding matrix for $\H_k$, only $\widetilde{\H}_k$ is needed.
	\subsection{Computational Complexity}
	In this subsection, we discuss the computational complexity of the proposed low complexity precoding (LCP) scheme and the WMMSE algorithm.
	For the proposed LCP algorithm, it is divided into three steps.
	The first step includes SVD operations and matrix multiplication of the merged channel, with a complexity of $\mathcal{O}\left(KN_t N_r^2 + KN_r^3 + N_t M^2\right)$, where $M$ denotes the number of the total data streams;
	the second steps includes the forward calculation of the network, with the complexity of $\mathcal{O}\left(\sum_{l=1}^{L-2}F_lJ_l^2C_{l-1}C_{l} + F_{L-1}F_{out}\right)$, where the $F_l$, $J_l$, $C_{l-1}$, $C_{l}$ denotes the feature map size, kernel size, number of input channel, number of output channel respectively, the computational complexity of offline training is omitted; the third steps includes the recovery operation in (\ref{eq:miso precoding matrix final presentation}), with a complexity of $\mathcal{O}\left(N_tM^2 + M^3\right)$.
	In conclusion, the computational complexity of the proposed LCP scheme is given by
	\begin{align}\label{eq:lcp flops}
		\mathcal{O}\left(KN_tN_r^2 + N_t M^2 + M^3 + KN_r^3 + \sum_{l=1}^{L-2}F_lJ_l^2C_{l-1}C_{l} + F_{L-1}F_{out}\right).
	\end{align}
{	According to \cite{shi2011iteratively} and \cite{hu2020iterative}, the computational complexity of WMMSE algorithm is given by:
	\begin{align}\label{eq:WMMSE flops}
		\mathcal{O}\left(L_w\left(K^2N_t N_r^2 + K^2N_t^2 N_r + KN_t^3 +K^2N_r^3\right)\right),
	\end{align}
	where $L_w$ denotes the number of iterations.}

	Compared (\ref{eq:lcp flops}) with (\ref{eq:WMMSE flops}), and with $M<\frac{1}{2}N_t$ in most scenarios, it can be found that by carefully setting the parameters of network, the computational complexity of the proposed LCP algorithm could be even less than one iteration of WMMSE algorithm. %, not to mention the number of iterations $L_w$.
	\section{Extention to MIMO-OFDM systems}
	\begin{figure}[tt]
		\centering
		\includegraphics[width=14cm]{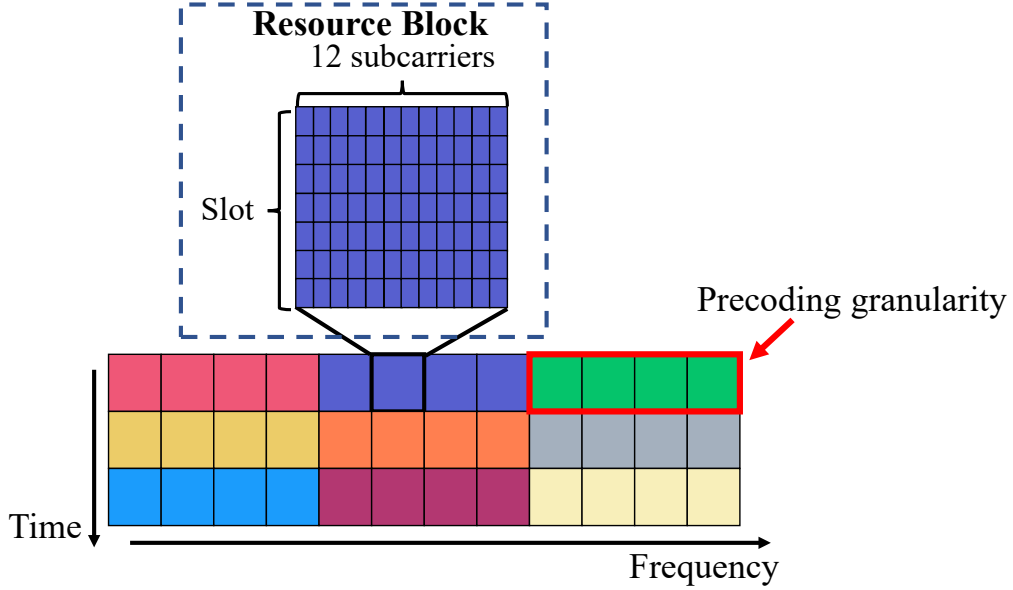}
		\caption{Precoding for the MIMO-OFDM system}
		\vspace{-7mm}
		\label{Fig:precoding for MIMO OFDM}
	\end{figure}
	In the practical OFDM system, as shown in Fig. \ref{Fig:precoding for MIMO OFDM},
	one resource block (RB) contains 12 adjacent subcarriers \cite{3gpp.36.331}, %, and the subcarriers channel in the same RB are assumed to be the same.
	%Out of the realistic reasons of the complexity and the correlation between adjacent RB channels,
	the same precoding matrix is used across a set of physical RBs.
	%the precoding granularity is usually not a single RB, but may be multiple RBs.
	The corresponding WSR maximization problem can be written as follows
	\begin{align}\label{pr:sum rate max}
		\mathscr{P}_2: &\max_{\left\{\V_{k}\right\}}\sum_{k=1}^{K}\sum_{b=1}^B \alpha_{k} R_{k,b},\\
		&\mathbf{s.t.}\ \sum_{k=1}^{K} {\rm Tr} (\mathbf{V}_k\mathbf{V}_k^H)\leq P_T. \notag
	\end{align}	
	where $K$ and $B$ denotes the number of users and precoding granularity respectively, the achievable rate of user $k$'s $b$-th RB channel  $R_{k,b}$ is given by
	\begin{align}\label{eq: R def}
		R_{k,b} = \log \det\left(\bf{I} + \H_{k,b} \V_k \V_{k}^H \H_{k,b}^H \left(\sum_{m\neq k}\H_{k,b} \V_m \V_m^H \H_{k,b}^H +\sigma^2 \bf{I}\right)^{-1} \right).
	\end{align}
	Different from $\mathscr{P}_1$, the $B$ RB channels $\left\{\mathbf{H}_{k,b},b=1,...,B\right\}$ share a common precoding matrix $\V_k$.
	Following the previous idea of dealing with $\mathscr{P}_1$, we first transform the original MIMO precoding problem to the MISO precoding problem, %through the way in Section \ref{subsec: miso 2 miso},
	then we derive the optimal solution structure for $\mathscr{P}_2$ in MISO scenario, as shown in the following proposition. %
%	We can still use the previous idea of dealing with $\mathscr{P}_1$ to get the low complexity solution of $\mathscr{P}_2$.
%	First we convert the original MIMO precoding problem to MISO precoding problem by using the converting way in Section \ref{subsec: miso 2 miso}.
	%Then we give the optimal solution structure of $\mathscr{P}_2$.
	\begin{proposition}\label{prop: multi RB optimal solution structure}
		\emph{For MU-MISO-OFDM system with the channels $\left\{\h_{k,b},\forall k,b\right\}$ of $M$ users, the optimal solution $\left\{\v_k^*\right\}$ of $\mathscr{P}_2$ can be expressed as}
		\begin{align}
			\v_k^*=\sqrt{p_k}f\left[\left(\sigma^2\mathbf{I}_{N_t} +\sum_{m=1}^M\sum_{b=1}^B \lambda_{m,b}\h_{m,b}^H\h_{m,b} \right)^{-1}\left(\sum_{b=1}^Bq_{k,b}\h_{k,b}^H\right)\right], \forall k,
		\end{align}
		{\emph{where $p_k$ and $\lambda_{k,b}$ meet $\sum_{k=1}^M p_k=\sum_{k=1}^M\sum_{b=1}^B \lambda_{k,b} = P_T$, $q_{k,b}\in \mathbb{C}$ is a weighted scalar for $\mathbf{h}_{k,b}^H$, $f\left[\cdot\right]$ is the normalization function. }}	
\end{proposition}
	\begin{proof}
		See Appendix B.
	\end{proof}
	\begin{remark}
		\emph{
		For WSR problem with the precoding granularity of multiple RBs, the optimal solution structure has three key features to be determined, e.g., $\mathbf{p}=[p_1,.,,,.p_M]$, $\boldsymbol{\lambda}=\left[\lambda_{1,1},...,\lambda_{M,B}\right]$,$\mathbf{q}=\left[q_{1,1},...q_{M,B}\right]$.
		The precoding design should first consider suppressing the interference between users, controlled by the feature $\boldsymbol{\lambda}$.
		{Then, since the channels of multiple RBs share a common precoding vector,
		for the sum rate maximization problem, the MRT beamformer is obtained by the weighted sum of the channels of each RB, with the weight factor $\mathbf{q}$.}
		Finally, similar to the optimal solution structure in (\ref{eq:miso optimal sulition structrue}) for $B=1$ scenario, the transmitted power is allocated according to $\mathbf{p}$.
		%The item determined the beamforming direction has two undetermined vector
		%explain the optimal solution structure's physical means. and compare it with the single RB's scenario.
		}
	\end{remark}
	With the optimal solution structure in Proposition \ref{prop: multi RB optimal solution structure}, we then give the deep learning based precoding scheme, as shown in Fig. \ref{Fig:network structure multi RB}. 	
	First, by conducting SVD for the MIMO channel, the corresponding MISO channel $\left\{\widehat{\h}_{k,b},\forall k,b\right\}$is obtained.
	Then, it still can be proved that the key features $\mathbf{p}$, $\boldsymbol{\lambda}$, $\mathbf{q}$ are the function of $\left\{\beta_{k},\widehat{\h}_{i,b}\widehat{\h}_{j,c}^H,\forall i,j,k,b,c\right\}$.
	Similar with the preprocessing in Section \ref{subsubsec: network input and output design}, let $\overline{\H} = \left[\sqrt{\beta_1}\widehat{\h}_{1,1},...,\sqrt{\beta_M}\widehat{\h}_{M,B}\right]$, we use $\overline{\H}^H\overline{\H}$  rather than $\overline{\H}$ as the input of the network.
	For the prediction task of $\mathbf{p}$ and $\mathbf{q}$,
	due to the increase of network input and output dimensions, we use more convolutional layers and fully connected layers to enhance the fitting ability.
	In addition, for the prediction of $\mathbf{q}$, as it is responsible for adjusting the MRT beamforming direction within the user, so the $\q_k$ is predicted separately through three fully connected layers, where the input $\mathbf{R}_k$ denotes $\overline{\H}_k^H\overline{\H}_k$, with $\overline{\H}_k=[\sqrt{\beta_k}\widehat{\h}_{k,1},...,\sqrt{\beta_k}\widehat{\h}_{k,B}]$.
	% So we compress the input by multiplying $\widehat{\H}^H$.
	% The the input's dimension should be $MB\times MB$.
	% So we use $\left\{p_k,\lambda_{k,b},q_{k,b}\right\}$ rather than $\left\{\h_{k,b},\forall k,b\right\}$ as the input of the network.

	% Then the  neural network should be designed carefully to learn the mapping from $\widehat{\H}^H\widehat{\H}$ to $\left\{p_k,\lambda_{k,b},q_{k,b}\right\}$.
	% Compared with the single RB precoding granularity scenario, the input's dimension increase from $M\times M$ to $MB\times MB$.
	% Moreover, there is an additional vector $\mathbf{q}$ needs to be predicted. Compared with $\mathbf{p}$ and $\boldsymbol{\lambda}$, $\mathbf{q}$ actually has no constraints, and thus hard to be predicted.
	% To answer these challenges, we first increase the number of convolutional layers, and then use two fully connected layers to output the key feature.
	% $\mathbf{p}$ and $\boldsymbol{\lambda}$ has the constant L1 norm, so we use the L1 normalization function to be the final activation function.
	% While $\mathbf{q}$ hasn't such a constraint, so we use the Tanh to be the final activation function.
	%Then the precoding
	%Then a similar Compared with the single RB precoding granularity scenario in Section \ref{sec:single RB deep learning based precoding algorithm}, the output's dimension is higher.
	\begin{figure}[tt]
		\centering
		\includegraphics[width=14cm]{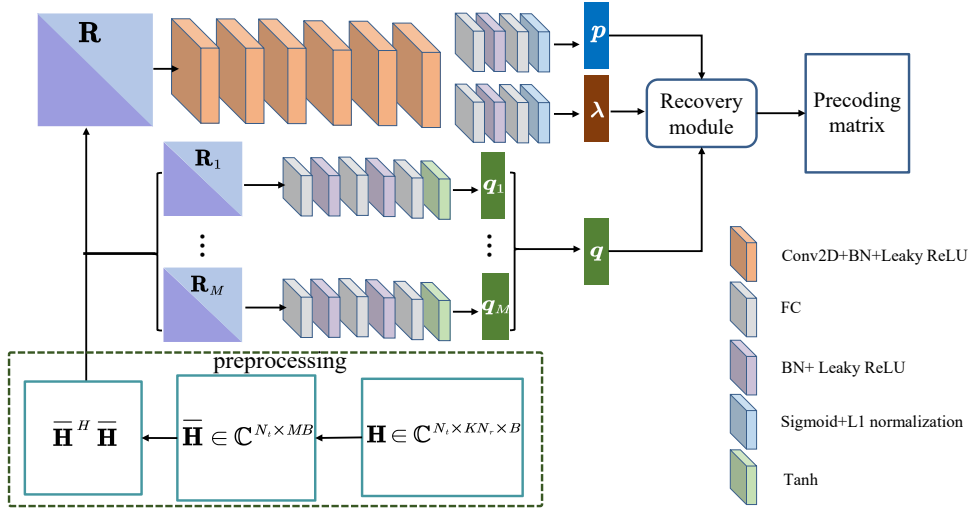}
		\caption{Network structure for MIMO-OFDM precoding}
		%\vspace{-10mm}
		\label{Fig:network structure multi RB}
	\end{figure}
	%Then to get the optimal solution structure, here we first give the corresponding WMMSE algorithm for $\mathscr{P}_2$.
	\section{Simulation results}
	\subsection{Simulation Setup}
	\textbf{System setup.} We consider the wideband mmWave MIMO-OFDM system in \cite{zhou2017low},  each resource block contains 12 adjacent subcarriers.
	The first subcarrier channel in each RB is viewed as the channel for the corresponding RB.
	%We assume that the subcarrier channels in one resource block are the same.
	The channel model of the $k$-th user on the $b$-th RB is given as follows.
	\begin{align}\label{eq: mimo ofdm channel model}
		\mathbf{H}_{k,b} = \sum_{l=1}^L \zeta_{k,l} e^{-j 2\pi \tau_{k,l} f_s b/B} \mathbf{a}_r\left(\theta_{k,l}\right)\mathbf{a}_t^T\left(\phi_{k,l}\right)
	\end{align}
	where $L$ is the number of path, $\zeta_{k,l}$ is the path gain of the $l$-th path, $f_s$ denotes sampling rate, $\tau_{k,l}$ denotes the path delay, $\theta_{k,l}$ and $\phi_{k,l}$ are the azimuth angles of arrival and departure (AoA/AoD) at the receiver and the transmitter, respectively.
	The simulation setting of the above parameters are given in Table \ref{table:channel parameters}.

	For the MIMO system setting, we consider two cases, including a standard case ($N_t=16$) and a massive case ($N_t=64$). 
	The detailed system settings are presented in Table \ref{table:mimo system parameters}. 
    {The number of filters and the kernel size for the three convolution layers in Fig. \ref{Fig:network structure single RB} are $\left\{16,7\right\}$, $\left\{8,5\right\}$, $\left\{4,3\right\}$. 
    The Adam optimizer is used for network parameter updating. 
    We adopt the  learning rate decay strategy, with an initial value of 0.01 at supervised learning stage and 0.001 at unsupervised learning stage.   
    }
	\begin{table}[]
		\centering
		\caption{System parameter}
	\begin{subtable}[Channel parameters]{
		%\caption{Channel parameters}
		%\centering
		\begin{tabular}{|c|c|}
		\hline
		Parameter & Value          \\ \hline
		$L$       & 10             \\ \hline
		$\tau_{k,l}$     & $\mathcal{U}[0,100] /{\rm ns}$ \\ \hline
		$\theta_{k,l},\phi_{k,l}$     & $\mathcal{U}[0,2\pi]$        \\ \hline
		%$\phi_{k,l}$       & $\mathcal{U}[0,2\pi]$        \\ \hline
		$f_s$        & 0.32 GHz        \\ \hline
		$\zeta_{k,l}$     & $\mathcal{CN}(0,1)$        \\ \hline
		\end{tabular}
		\label{table:channel parameters}}
	\end{subtable}
	\qquad
	\begin{subtable}[System configuration]{
		%\centering
		%\caption{System parameter of two cases}
		\begin{tabular}{|c|c|c|}
		\hline
		Parameters & case 1       & case 2       \\ \hline
		$N_t$               & 16           & 64           \\ \hline
		$N_r$               & 2            & 4            \\ \hline
		$K$                & 4$\sim$8          & 8$\sim$16         \\ \hline
		$d_k$               & 1            & 2            \\ \hline
		SNR ($\frac{p}{\sigma^2}$)             & 0 dB & 0 dB \\ \hline
		\end{tabular}
		\label{table:mimo system parameters}}
		\end{subtable}
	\end{table}

	\textbf{Benchmark schemes.} We compare the proposed {low complexity precoding} (LCP) scheme with three benchmark schemes including: 1) high performance baseline, WMMSE algorithm \cite{shi2011iteratively}; 2) data-driven baseline, \emph{learning to $\mathbf{u}$ and $\mathbf{w}$} (LUW) algorithm \cite{hu2020iterative,lu2020learning}, which achieves a good tradeoff between performance and complexity;
	3) low complexity baseline, eigen-based zero-forcing (EZF) algorithm, which obtains the precoding matrix by the Moore-Penrose pseudo-inverse of $\widehat{\H}$ defined in (\ref{eq:miso channel merge definition}).
%	Here mention the three benchmark scheme, including WMMSE,EZF, and the LUW algorithm.
	\subsection{Problem Transformation Loss}\label{subsec: problem transformation loss}
	%In this section we will verify the loss caused by problem transformation in Algorithm \ref{algo:sampling with relpacement} is small and can be neglected in most practical scenarios.
	Under the typical system setting in case 2, the performance loss caused by Algorithm  \ref{algo:sampling with relpacement} in different transmission SNR and number of users is shown in  Table \ref{table: performance loss test},
	where WMMSE denotes the performance of WMMSE algorithm, and LCP (ideal) denotes the optimal performance solved by Algorithm \ref{algo:sampling with relpacement} (obtained by running WMMSE algorithm to solve the corresponding MISO precoding problem in step 6), it can be viewed as the performance of the proposed LCP algorithm when the neural network is able to fit the mapping between the channel and key features perfectly.
	First, it can be seen that in all the simulation setting, the performance loss caused by the problem transformation is no more than 3\%. %and thus can be neglected in most scenarios.
	Second, as the SNR increases, the performance loss increases slightly. This is because Algorithm \ref{algo:sampling with relpacement} discards some available channels with small channel gains, which can still be used to transmit some information when the channel noise is small enough.
	%Nevertheless, the performance loss is still in an acceptable range even at a high SNR.
	Third, an increase in the number of users will not have a significant impact on the performance loss.
	It can be concluded that the performance loss caused by problem transformation in Algorithm \ref{algo:sampling with relpacement} is small enough and can be neglected in most practical scenarios.
	\begin{table}[]
		\centering
		\caption{Performance loss caused by problem transformation}\label{table: performance loss test}
	\begin{subtable}[case 2, $K=10$, SNR = $-5\sim 15$dB]{
		%\caption{Channel parameters}
		%\centering
		\begin{tabular}{llllll}
			\hline
			SNR (dB)        & -5               & 0               & 5               & 10              & 15               \\ \hline
			WMMSE      & 24.921         & 44.325        & 70.317        & 100.550       & 132.320        \\ \hline
			LCP (ideal) & 24.715 (99.2\%) & 43.601 (98.4\%) & 68.505 (97.4\%) & 97.550 (97.0\%) & 128.822 (97.3\%) \\ \hline
		\end{tabular}
		\label{subtable: differ SNR for test converting loss}}
	\end{subtable}
	\qquad
	\begin{subtable}[case 2, SNR$=0$dB, $K=8\sim 16$]{
		%\centering
		%\caption{System parameter of two cases}
		\begin{tabular}{llllll}
			\hline
			{\# of user}       & 8               & 10              & 12              & 14              & 16              \\ \hline
			WMMSE      & 40.242          & 44.325          & 47.658          & 50.458          & 52.843          \\ \hline
			LCP (ideal) & 39.643 (98.5\%) & 43.601 (98.4\%) & 46.836 (98.3\%) & 49.561 (98.2\%) & 51.877 (98.2\%) \\ \hline
		\end{tabular}
		\label{table:mimo system parameters}}
		\end{subtable}
	\end{table}
	\subsection{Performance and Complexity Evaluation}
%	In this part, we first verify the proposed precoding algorithm under single RB granularity in Section \ref{sec:single RB deep learning based precoding algorithm}.
	\subsubsection{Sum-Rate Performance}
	\begin{figure}[tt]
		\centering
	  \subfigure[case 1]{
		\label{Fig:single RB differ user Nt 16:a} %% label for first subfigure
		\includegraphics[width=0.485\linewidth]{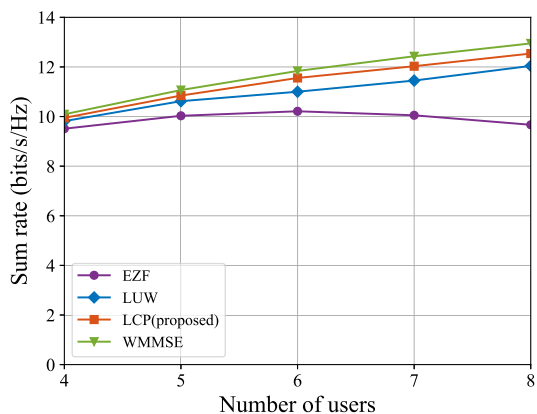}}
	  \subfigure[case 2]{
		\label{Fig:single RB differ user Nt 64:b} %% label for first subfigure
		\includegraphics[width=0.485\linewidth]{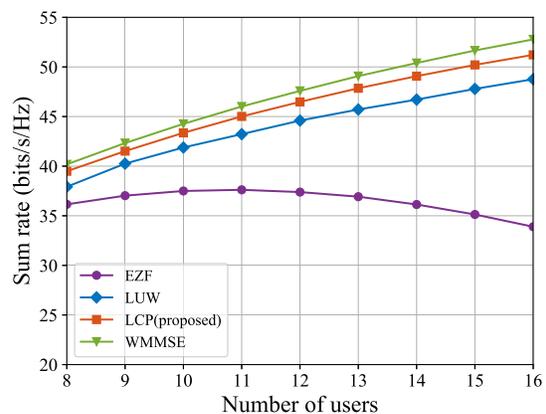}}
	
		\caption{Performance comparison in two system cases}\label{fig: performance comparison in 1RB}
		%\vspace{-10mm}
		\label{Fig:single RB differ user Nt}
	\end{figure}
	In this subsection, we evaluate the performance of the proposed LCP algorithm.
	The sum rate performance under two system cases in Table \ref{table:mimo system parameters} are illustrated in Fig \ref{fig: performance comparison in 1RB}.
	Under the system setting of case 1, the sum rate performance of the proposed LCP scheme has a similar performance with WMMSE algorithm, then followed by LUW and EZF algorithm.
	With the increase of the number of users, the performance gap between LCP and WMMSE algorithm increases, the same trend also appears in the system setting of case 2.
	This is because with the increase of $K$, the input-output dimension of the network increases, resulting in the increase of the loss caused by network fitting.
	But the gap still holds in an acceptable range. 	
	Similarly, as the LUW algorithm has a larger output, the performance gap between the LUW algorithm and the proposed LCP algorithm also increases with the increase of $K$.
	As for EZF algorithm, due to the limitation of solution space dimension, its performance is significantly worse than other algorithms in the scenario of large $K$.
	In conclusion, the proposed LCP scheme not only outperforms the LUW and EZF algorithm, but also achieves comparable performance to WMMSE algorithm.

	% In case 1's system setting, the above four precoding schemes have similar performance, the iterative WMMSE algorithm achieve the best performance, followed closely by the proposed LCP algorithm.
	% As the number of user increases, LUW scheme and EZF scheme has a wider performance gap with the LCP scheme compared with the system setting of less user's number.
	% A similar trend can be seen in the massive system setting, case 2.
	% Due to the heuristic nature, the ZF scheme has an obvious poor performance.
	% The proposed LCP algorithm not only outperforms the LUW algorithm, but still followed the WMMSE's performance well.
	%The network

	\subsubsection{Computational Complexity}
	\begin{figure}[tt]
		\centering
	  \subfigure[case 1]{
		\label{Fig:complexity:a} %% label for first subfigure
		\includegraphics[width=0.485\linewidth]{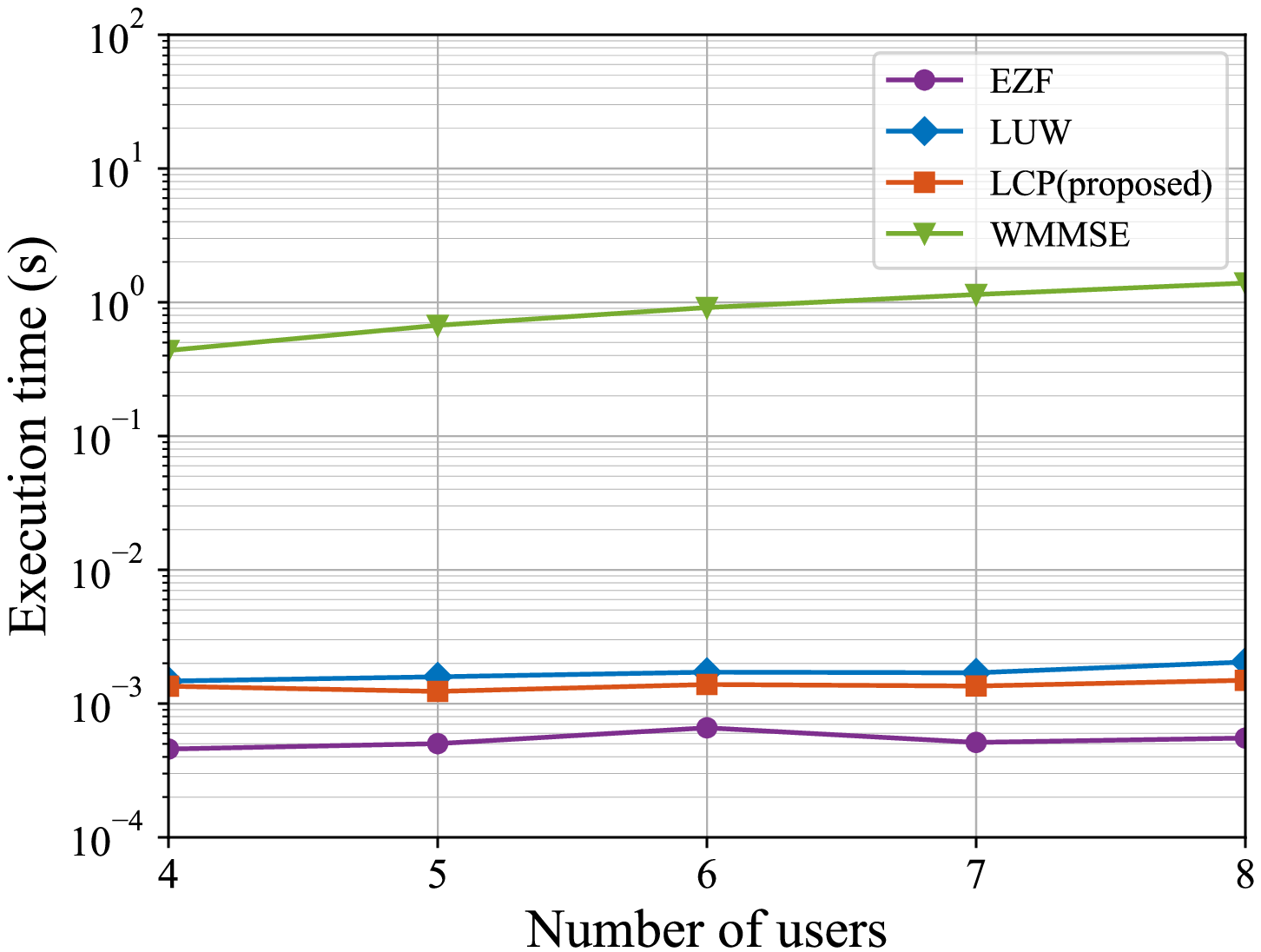}}
	  \subfigure[case 2]{
		\label{Fig:complexity:b} %% label for first subfigure
		\includegraphics[width=0.485\linewidth]{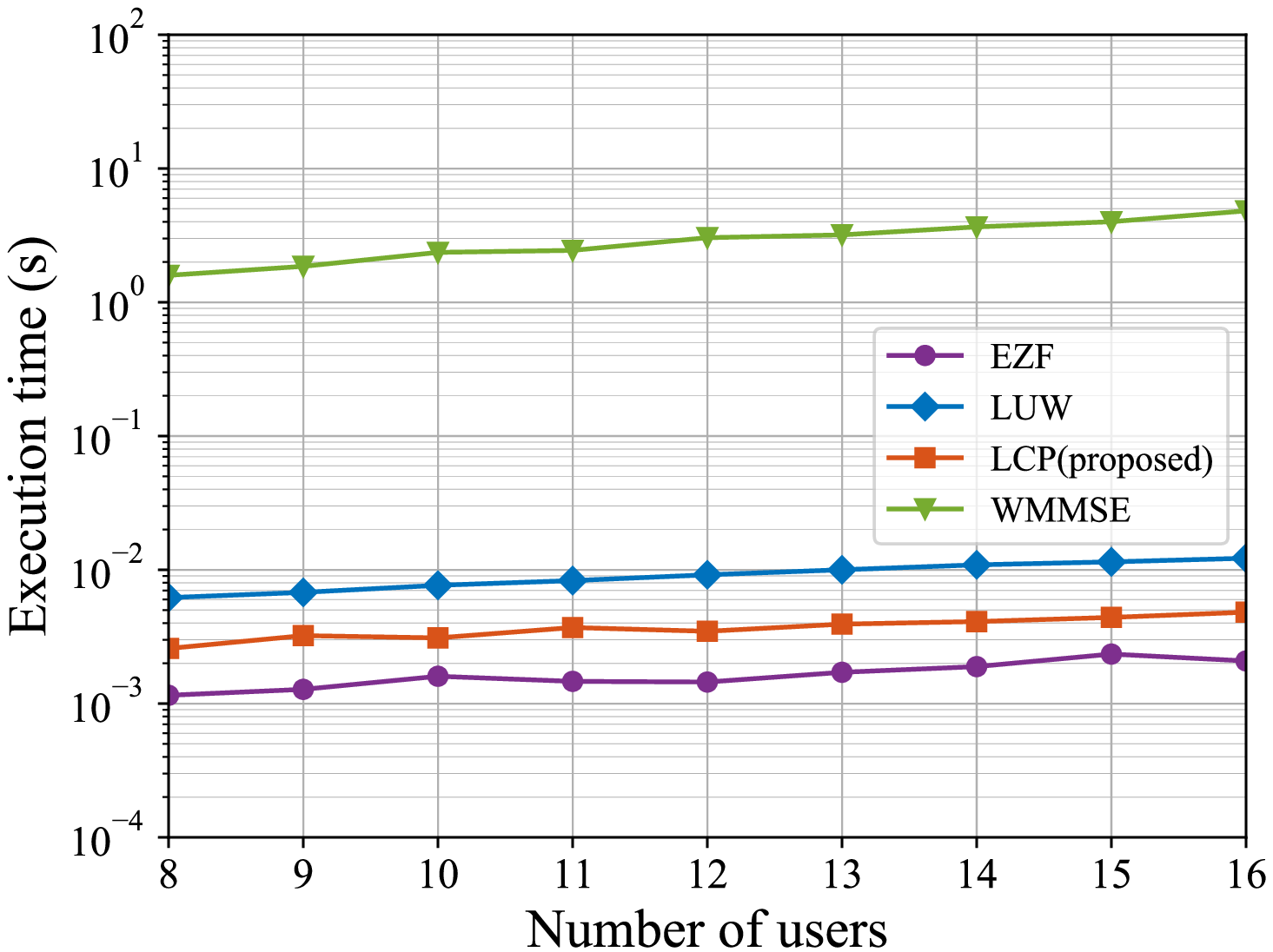}}
	
		\caption{{Execution time comparison in two system setting cases}}
		%\vspace{-10mm}
		\label{Fig:complexity}
	\end{figure}
	Fig \ref{Fig:complexity} presents the CPU processing time of the four precoding schemes in the prediction stage.
	%where %WMMSE is assumed to stopped when the iteration variable $\mathbf{W}$ meets $\sum_{k=1}^K \left\|\mathbf{W}_{k}^{now}-\mathbf{W}_{k}^{before}\right\|\leq 10^{-4}$,
	%and running time of the two learning based algorithm (e.g., LUW and LCP) in the prediction stage.
	%EZF (2X) denotes twice the execution time of EZF.
	In the small system setting case, the execution time of the proposed LCP algorithm is slightly less than that of the LUW algorithm, and slightly more than twice that of EZF.
	Reasonably, due to the inversion of high-dimensional matrix in each iteration and multiple iteration,  WMMSE algorithm takes much longer time to complete the iterative process, preventing it to be applied in the real-time communication systems.
	A similar phenomenon can also be found in the large system setting case.
	However, as the computational complexity of the recovery module in LUW is $	\mathcal{O}\left(KN_t N_r^2 + KN_t^2 N_r + KN_t^3 +KN_r^3\right)$, the execution time of LUW algorithm is clearly increased.
	Fortunately, the proposed LCP algorithm still holds its execution time about twice that of EZF algorithm. %, and achieves near performance of WMMSE.
	In conclusion, the proposed LCP algorithm achieves similar performance of WMMSE algorithm with only twice execution time of EZF algorithm, and has better performance of LUW algorithm with smaller execution time,
	%which shows that the proposed LCP algorithm is computation-efficient.
	which shows its superiority in terms of both performance and complexity.

	%it is more valuable in the recent popular massive MU-MIMO system.
	%In this part, we introduce the complexity of the proposed LCP scheme and the three benchmark schemes.
	
	\subsection{Impact of Network Pruning}
	\begin{table}[]
		\centering
		\caption{Model pruning results}\label{table: model pruning}
		\vspace{3mm}
		\begin{tabular}{c|cc|c|c}
				\hline
				\multirow{2}{*}{\#users (K)} & \multicolumn{2}{c|}{\# of pruning filters } & \multirow{2}{*}{sum rate} & \multirow{2}{*}{relative complexity of the network part} \\ \cline{2-3}
				& \multicolumn{1}{c|}{  first layer} & \multicolumn{1}{c|}{  second layer} &          &          \\\hline
				\multirow{4}{*}{10} & \multicolumn{1}{c|}{0}           & \multicolumn{1}{c|}{0}            & 43.34 & 1.00 \\
				& \multicolumn{1}{c|}{8}           & \multicolumn{1}{c|}{4}            & 43.30 & 0.35 \\
				& \multicolumn{1}{c|}{12}         & \multicolumn{1}{c|}{6}            & 43.26 & 0.14 \\
				& \multicolumn{1}{c|}{15}          & \multicolumn{1}{c|}{7}           & 43.13 & 0.05 \\ \hline
			\end{tabular}%
		\vspace{-3mm}
	\end{table}
	Table \ref{table: model pruning} shows the performance and the complexity when removing different number of convolution filters. Firstly the large model ($J_1=16$, $J_2=8$, $J_3 = 4$) is carefully trained, and then different size of model is obtained by pruning convolution filters recursively.
	With more convolution filters being pruned, the sum rate only decreases slightly,  which indicates there are some redundant parts in the network.
	%And by removing the them, the computational complexity could be significantly decreased.
	Moreover, according to the complexity requirements of different systems, the dynamic adjustment between complexity and performance can be realized by cutting off different number of convolution filters. 	
	\subsection{Generalization}
	\subsubsection{Varying Number of User Streams}
	Fig \ref{Fig:differ ds} shows the performance of users with the same or different number of data streams under case 2's system setting,
	as dicussed in Section \ref{subsec: vary ds}, the total number of streams is set to $2K$, for the same $d_k$ case, the number of data streams per user $d_k$ is fixed to $2$; for the different $d_k$ case, $d_k$ could vary from $1$ to $4$ between users.
%	where LCP (same) and WMMSE (same) respectively denotes the performance of the two algorithms when the number of data streams per user $d_k$ is 2, while LCP and WMMSE respectively denotes the performance of the two algorithm when $d_k$ varies in $1$-$4$, but the total number of streams equals to $2K$.
	First, compared with the same $d_k$ case,  both algorithms have a worse performance when $d_k$ is different between users.
	For a specific user, allocating too many data streams will not increase the sum rate, while too few will have a great impact on the sum rate performance since the channel between the base station and the user is not fully utilized.
	%Therefore, the sum rate performance with different data streams is worse than that with the same data streams.
	Second, it is found that the performance gap between WMMSE and LCP in the same $d_k$ setting (blue shadow) and the gap in the different $d_k$ setting (red shadow) are almost the same, which shows that even in the scenario with dynamic changes of user streams, with the transformation in Algorithm \ref{algo:sampling with relpacement}, the proposed scheme can still maintain a performance close to WMMSE algorithm without any change to the network design.
	\begin{figure}[tt]
		\centering
		\includegraphics[width=10cm]{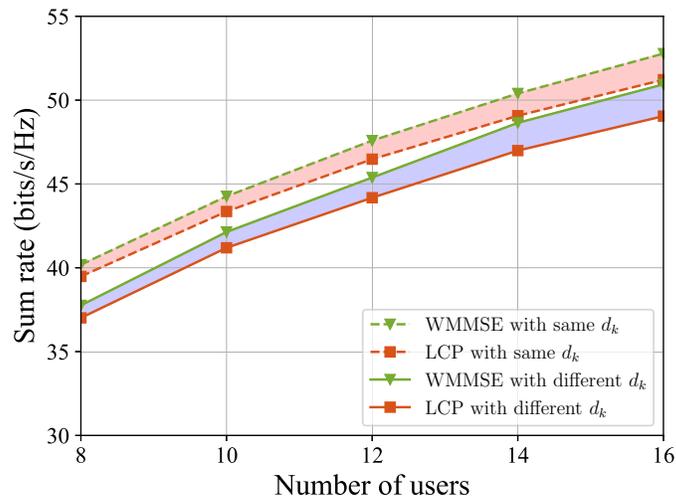}
		\caption{Performance comparison in the scenario of users with same or different $d_k$ }
		%\vspace{-10mm}
		\label{Fig:differ ds}
	\end{figure}

	\subsubsection{Varying Number of Users}
	Fig \ref{Fig:zero patch} presents the user number generalization performance of the proposed method under different training settings,
	%Fig \ref{Fig:zero patch} presents the generalization ability of the proposed method,
	where the abscissa represents the actual number of users in the system.
	For a specific number of users $K_0$,
	three training settings are considered, including 1) trained with data of specific $K$: training a $K=K_0$ model with $\left\{\H_k, k\leq K_0\right\}$, testing on $\left\{\H_k, k\leq K_0\right\}$;
	2) trained without zero filling data: training a $K=12$ model with $\left\{\H_k, k\leq 12\right\}$, testing on  $\left\{\H_k, k\leq K_0\right\}$ and $\left\{\H_k = \mathbf{0}, K_0<k \leq 12\right\}$;
	3) trained with zero filling data: training a $K=12$ model with $\left\{\H_k, k\leq K_0\right\}$ and $\left\{\H_k = \mathbf{0}, K_0<k \leq 12\right\}$, testing on  $\left\{\H_k, k\leq K_0\right\}$ and $\left\{\H_k = \mathbf{0}, K_0<k \leq 12\right\}$.
	%	LCP (customized) denotes the performance of a well trained network with $K=K_0$ with channel dataset of $K_0$ users,
%	then the corresponding channel matrix of 12 users is obtained by filling in zero for redundant $12-K_0$ user channels,
%	we further test the performance of the WMMSE algorithm and the proposed LCP algorithm on the zero filling channel data.
	%which could be viewed as the ideal generalization performance of the proposed method.
%	the test performance of the WMMSE algorithm is obtained
%	WMMSE and LCP (customized) denotes the test performance of WMMSE algorithm and the proposed method in specific number of users.
%	We first obtain the corresponding channel matrix of 12 users by filling in zero for redundant user channels,
%	and then test the performance of WMMSE algorithm and LCP algorithm on such a channel data.
	For the WMMSE algorithm, it can naturally adapt to the zero filling input, so it maintains nearly optimal performance, and thus can be used as a performance baseline of generalization ability.
	For the proposed algorithm,
	%two offline training strategy are considered, that is, training without zero filling data and training with zero filling data.
	when the training data does not contain zero filling data, the network can not adapt to the input of zero filling data, the performance decreased obviously.
	Fortunately, with the help of zero filling data for training,
	it can be seen that
	the $K=12$ model trained on zero filling data  achieves almost the same performance with the customized model trained on specific data, and a close performance to the WMMSE algorithm,  which shows its strong generalization ability on number of users.
	%it can be seen that the proposed scheme has a close performance with WMMSE algorithm.
	In conclusion, the data-driven method should try to make the training data and test data have similar statistical characteristics.
	Once the proposed scheme has partial zero filling data samples for auxiliary training, the generalization ability can be greatly improved.
	
	\begin{figure}[tt]
		\centering
		\includegraphics[width=10cm]{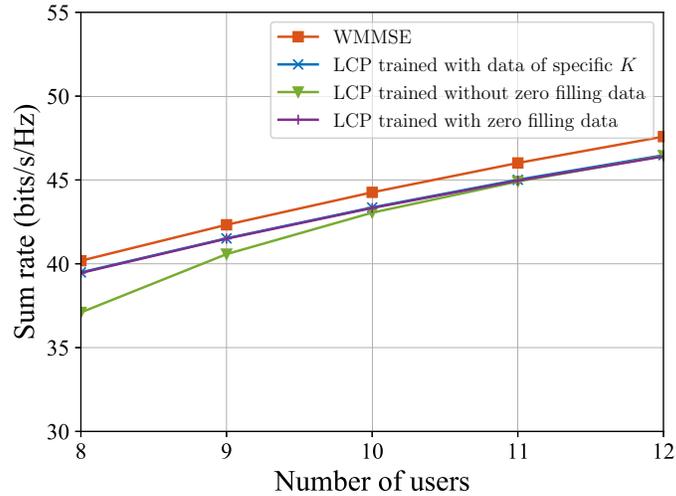}
		\caption{The generalization ability comparison in different training settings}
		\vspace{-8mm}
		\label{Fig:zero patch}
	\end{figure}

	\subsubsection{Imperfect CSI}
	Fig. \ref{Fig:imperfect CSI} shows the performance of the four precoding schemes with imperfect CSI,
	where the precoding matrix is designed based on the noisy channel, i.e., $\widetilde{\H} = \H + \mathbf{n}$, and the actual sum rate is calculated based on the noiseless channel $\H$.
%	where the channel estimation SNR denotes ${\left\|\mathbf{H}\right\|^2}/{\left\|\mathbf{n}\right\|^2}$, with the noisy channel $\widetilde{\H} = \H + \mathbf{n}$.
	For WMMSE algorithm and EZF algorithm, %due to their strong dependence on the accurate channel matrix,
	their performance drops a lot due to the strong dependence on the accurate channel.
	For the proposed LCP algorithm, it can adapt to such noise through a large number of training samples in the offline training stage, so as to performance better.
	Though the LUW algorithm, another learning based scheme, enjoys such adaptive process as well, the proposed LCP algorithm outperforms the LUW algorithm in all the error settings.
	With the decrease of channel estimation error, the channel becomes relatively more accurate, the WMMSE algorithm leads the performance again at the cost of huge computational complexity, but still the proposed LCP algorithm has a close performance with the WMMSE algorithm, which demonstrates the robustness of the proposed LCP algorithm in imperfect CSI scenarios.
	%where the four schemes obtain the precoding matrix through the noisy channel (i.e., $\widetilde{\H} = \H + \mathbf{n}$), and the sum-rate is calculated according to the noiseless channel $\H$.
	\begin{figure}[tt]
		\centering
		\includegraphics[width=10cm]{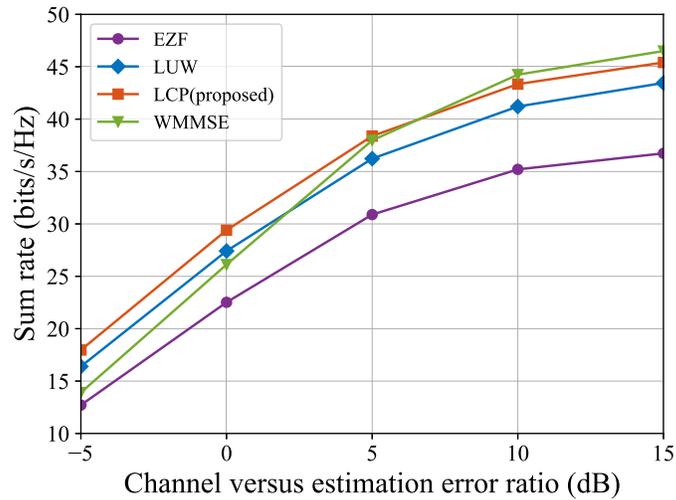}
		\caption{The sum-rate performance in the scenario of imperfect CSI}
		\vspace{-5mm}
		\label{Fig:imperfect CSI}
	\end{figure}
	\subsection{Precoding with Granularity of Multi-RB Channels}
	\begin{figure}[tt]
		\centering
	  \subfigure[case 1]{
		\label{Fig:differ_SNR:a} %% label for first subfigure
		\includegraphics[width=0.485\linewidth]{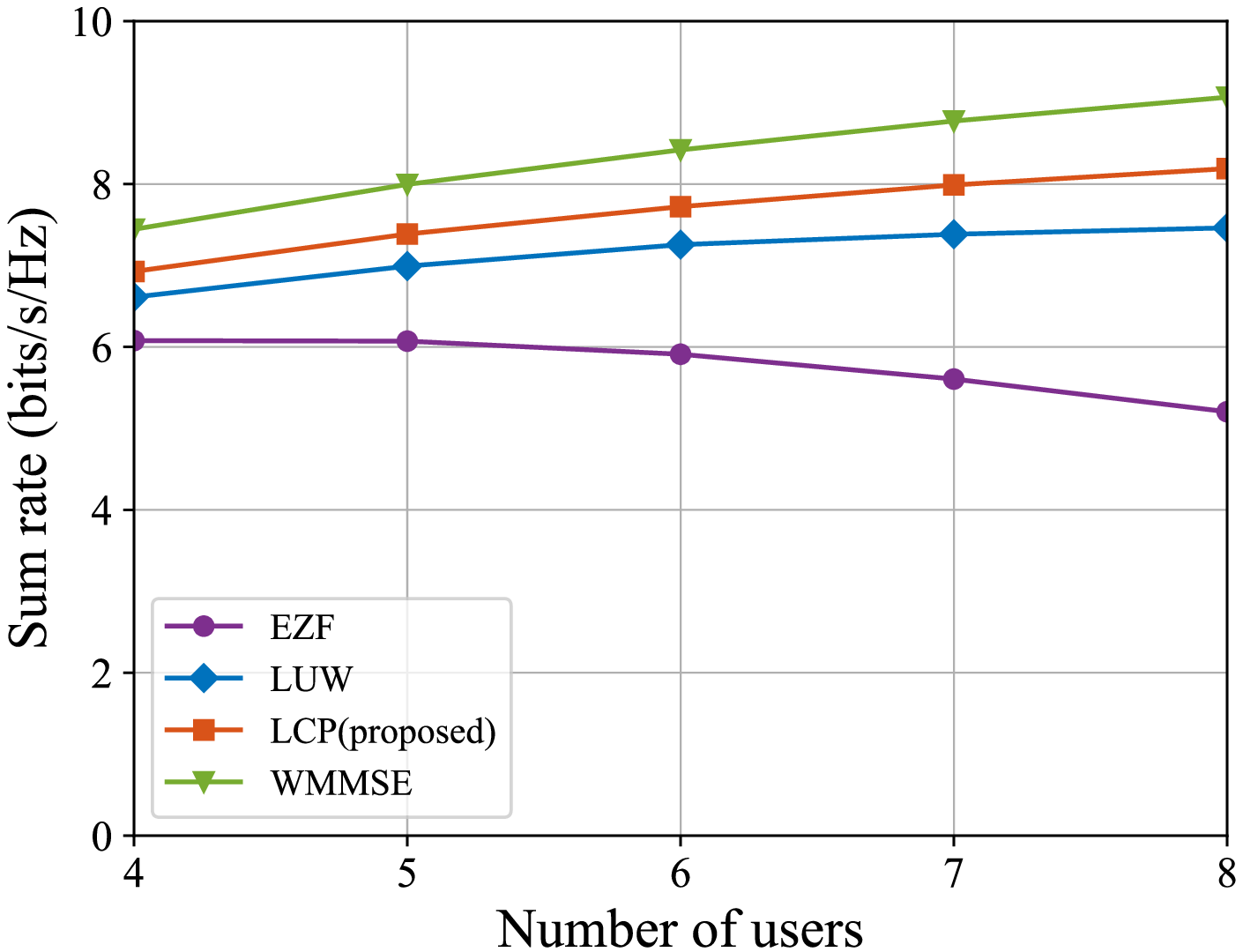}}
	  \subfigure[case 2]{
		\label{Fig:differ_SNR:b} %% label for first subfigure
		\includegraphics[width=0.485\linewidth]{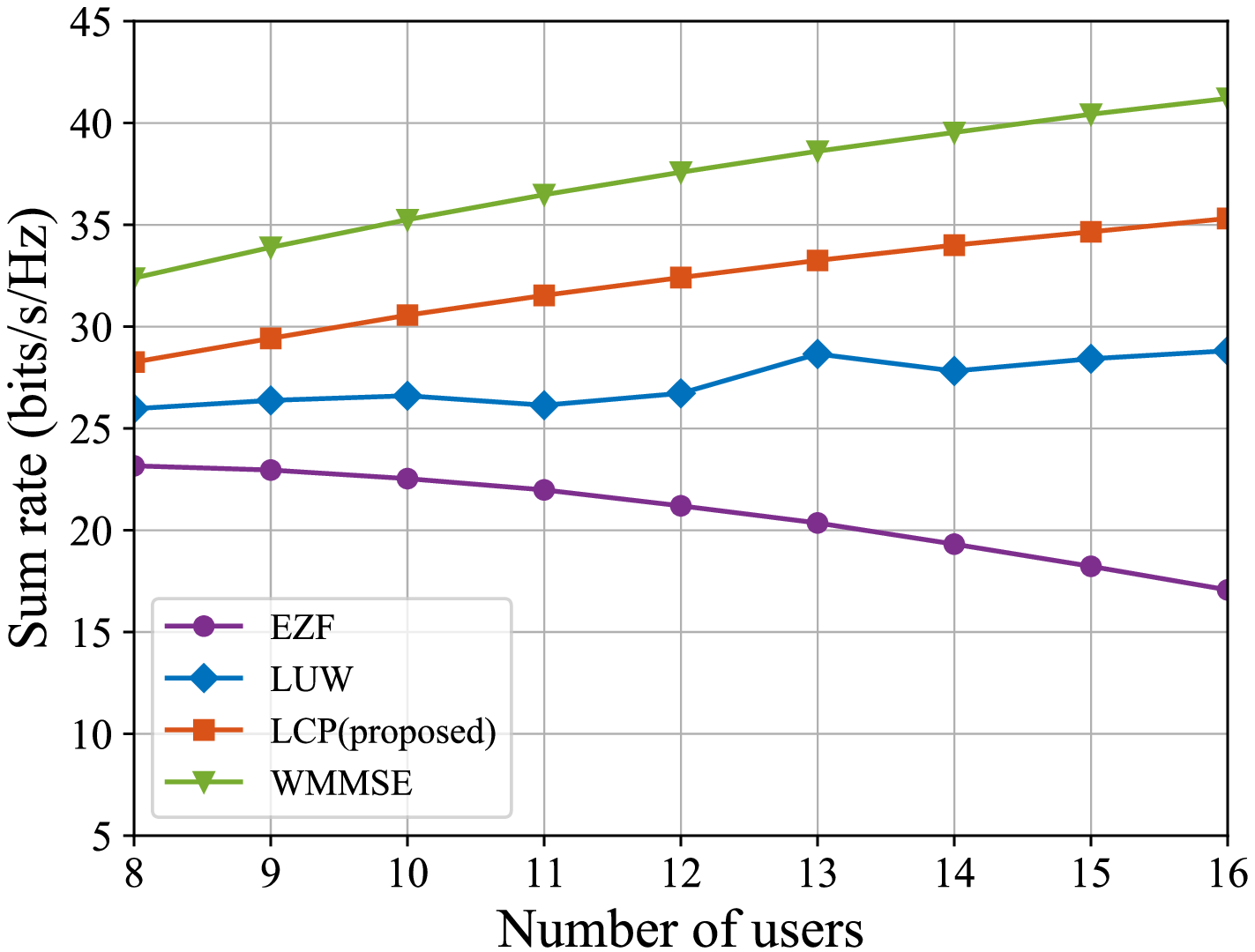}}
	
		\caption{{Performance comparison in the MIMO-OFDM system}}
		\vspace{-5mm}
		\label{Fig:differ user B = 4}
	\end{figure}
	Fig \ref{Fig:differ user B = 4} shows the sum rate performance of the four precoding schemes when the precoding granularity $B=4$.
	It can be found that in the two system setting cases, the WMMSE algorithm,
	%based on optimization theory,
	still leads the performance.
	For  the proposed LCP algorithm, it achieves nearly 90\% performance of WMMSE algorithm, while has much lower computational complexity.
	Moreover, compared with EZF and LUW algorithm which has a similar computational complexity, LCP algorithm significantly outperforms them in the sum-rate performance.
	It can be concluded that the proposed LCP algorithm has good practicability in MIMO-OFDM system.
	\begin{figure}[tt]
		\centering
		\includegraphics[width=10cm]{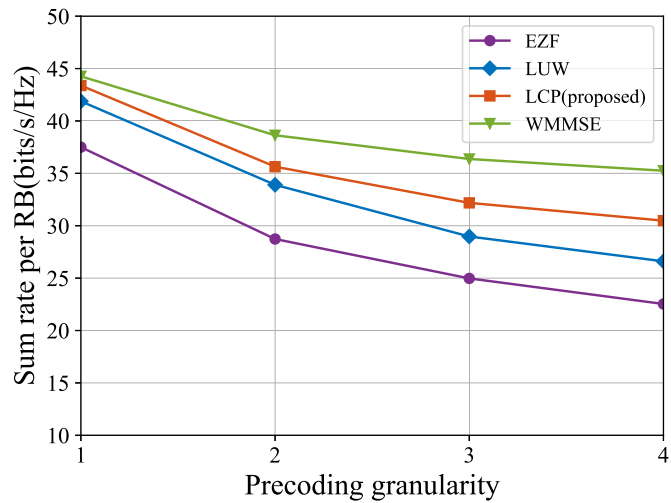}
		\caption{Performance comparison in different precoding granularity}
		\vspace{-10mm}
		\label{Fig:differ RB}
	\end{figure}
	Fig \ref{Fig:differ RB} shows the sum rate performance of the four schemes in different precoding granularity, where the sum rate is divided by the precoding granularity $B$.
	When $B$ increases, there are more RB that needs to share a common precoding matrix,
	the precoding algorithm needs to take all of them into account,  and thus the per RB's sum-rate decreases.
	It can be found that the proposed LCP algorithm outperforms the LUW algorithm and the EZF algorithm a lot, and still has the  closest performance to WMMSE algorithm.
	The performance gap between the proposed LCP algorithm and WMMSE algorithm increases with the increase of $B$.
	This is because when $B$ increases, the key feature $\left\{q_{k,b}\right\}$  has a greater impact on the performance, and the learning performance of neural network drops due to  the more complex mapping from the channel to $\mathbf{q}$.
	This problem can be solved by designing a more effective network, which is worthy of further study to promote the application of the proposed algorithm in the practical system.

	%Even with the above difficulties, the proposed algorithm outperforms the other two computation efficient algorithm
	%does not affect the recovery of precoding matrix, so as to the sum-rate performance, which has discussed in \cite{bjornson2014optimal}.
	%While for $B\geq2$, in (\ref{prop: multi RB optimal solution structure}), the key feature $\left\{q_{k,b}\right\}$ is a scalar for the weighted sum of each resource block's channel, and has a great impact on the performance.
	%The network is hard to learning the mapping from channel to the three key features, which introduces more performance loss.

	\section{Conclusion}
	\vspace{-2mm}
	In this paper, we have proposed a novel data-driven precoding algorithm.
	Firstly, we develop a method to transform the MIMO precoding problem into MISO precoding problem.
	Then, we present a DL-based solution for the corresponding MISO precoding problem, and further reduce the complexity through recovery module compression and model pruning.
	The extension to the MIMO-OFDM systems is also discussed.
	Experimental results show that the proposed scheme achieves similar performance with the WMMSE algorithm, and has a much lower computational cost.
	For future work, it is interesting to extend the proposed scheme to hybrid precoding  scenarios and multi-cell systems.

	% In this work, we proposed a low complexity AI-based precoding algorithm.
	% We first convert the MIMO precoding problem to MISO precoding for simplification,
	% then based on the relationship of two problem's optimal solution, the solution's structure has been found.
	% Then we resort to deep learning to tackling a nonconvex optimization problem.
	% Experiments shows that the proposed scheme achieves nearly optimal performance with significantly low complexity.

	\appendix
	\subsection{Proof of Proposition \ref{proposition:input instruct}}\label{subsec: proof of proposition 2}
	\begin{proof}
		From (\ref{eq:optimal solution's structrue without normalization}), it can be found that $\left\{p_k,\lambda_k\right\}$ is determined by the final $\left\{u_k^N,w_k^N\right\}$, where $N$ denotes the number of iterations to convergence.
		Therefore, to establish proposition 2, we only needs to prove that $\left\{u_k,w_k\right\}$ is determined by $\left\{\mathbf{h}_i\mathbf{h}_j^H\right\}$.
		For convenience of expression, we use $x \rightarrow y$ to denote that $y$ is a function of $x$.

		%So it is only necessary to prove that $\left\{u_k,w_k\right\}$ is determined by $\left\{\mathbf{h}_i\mathbf{h}_j^H\right\}$.
		With the intermediate variable $\left\{\mathbf{v}_k^{t-1},u_k^{t-1},w_k^{t-1}\right\}$,
		the WMMSE algorithm obtains $\left\{\mathbf{v}_k^{t},u_k^{t},w_k^{t}\right\}$ by performing the following three steps.
		%The WMMSE iterative algorithm in the MISO scenario is rewritten as follow.
		\begin{align}\label{eq:wmmse miso u update}
			\mathbf{v}_k^{t}&=\left(\sigma^2\mathbf{I}_{N_t}+\sum_{m=1}^M \lambda_m^{t-1} \mathbf{h}_m^H\mathbf{h}_m\right)^{-1}\mathbf{h}_k^H \gamma_k^{t-1}, \forall k,\\
			u_k^t&=\left(\frac{\sigma^2}{P_T}\sum_{m=1}^M \left(\mathbf{v}_m^t\right)^H\mathbf{v}_m^t+\mathbf{h}_k\mathbf{v}_m^t\left(\mathbf{v}_m^t\right)^H\mathbf{h}_k^H\right)^{-1}\mathbf{h}_k\mathbf{v}_k^t,\forall k,\label{eq: wmmse ut iteration}\\
			w_k^t&=\left(1-(u_k^t)^H\mathbf{h}_k\mathbf{v}_k^t\right)^{-1},\forall k.	\label{eq: wmmse wt iteration}	
		\end{align}
		where $\lambda_m^{t-1}=P_T\frac{\alpha_m u_m^{t-1}w_m^{t-1}\left(u_m^{t-1}\right)^H}{\sum_{n=1}^M \alpha_n u_n^{t-1}w_n^{t-1}\left(u_n^{t-1}\right)^H}$, $\gamma_m^{t-1}=P_T\frac{\alpha_mu_m^{t-1}w_m^{t-1}}{\sum_{n=1}^M \alpha_n u_n^{t-1}w_n^{t-1}\left(u_n^{t-1}\right)^H}$.
	
		Let $\boldsymbol{\Lambda}^{t}\triangleq{\rm diag}\left(\lambda_1^{t},...,\lambda_K^{t}\right)$, $\mathbf{\Gamma}^{t}\triangleq{\rm diag}\left(\gamma_1^t,...,\gamma_K^t\right)$, $\mathbf{H}\triangleq \left[\mathbf{h}_1^H,..,\mathbf{h}_M^H\right]$, $\mathbf{V}^{t}\triangleq \left[\v_1,...,\v_M\right]$.
		%Based on (\ref{eq:wmmse miso v update}), let $\mathbf{V}\triangleq \left[\v_1,...,\v_M\right]$,
%		$\mathbf{H}\triangleq \left[\mathbf{h}_1^H,..,\mathbf{h}_M^H\right]$, $\boldsymbol{\Lambda}\triangleq{\rm diag}\left(\lambda_1,...,\lambda_M\right)$, $\mathbf{A}\triangleq {\rm diag}\left(a_1,...,a_M\right)$.
		Then $\mathbf{V}^{t}$ can be expressed as
		\begin{align}\label{eq:v merge expression init}
			\mathbf{V}^{t}=\left(\sigma^2\mathbf{I}_{N_t}+\mathbf{H}\boldsymbol{\Lambda}^{t-1}\mathbf{H}^H\right)^{-1}\mathbf{H}\boldsymbol{\Gamma}^{t-1}=\left(\sigma^2\mathbf{I}_{N_t}+\widetilde{\mathbf{H}}\widetilde{\mathbf{H}}^H\right)^{-1}\widetilde{\mathbf{H}}\left(\boldsymbol{\Lambda}^{t-1}\right)^{-\frac{1}{2}}\mathbf{\Gamma}^{t-1},		\end{align}
		where $\widetilde{\mathbf{H}}=\mathbf{H}\left(\boldsymbol{\Lambda}^{t-1}\right)^{\frac{1}{2}}$.
		%Let $\widetilde{\mathbf{H}}=\mathbf{H}\boldsymbol{\Lambda}^{\frac{1}{2}}$, then we have

		Then we have
		\begin{align}\label{eq:v merge expression reformulate}
			\mathbf{V}^t&=\left(\sigma^2\mathbf{I}_{N_t}+\widetilde{\mathbf{H}}\widetilde{\mathbf{H}}^H\right)^{-1}\widetilde{\mathbf{H}}\left(\boldsymbol{\Lambda}^{t-1}\right)^{-\frac{1}{2}}\mathbf{\Gamma}^{t-1}\\
			&=\widetilde{\mathbf{H}}\left(\sigma^2\mathbf{I}_{M}+\widetilde{\mathbf{H}}^H\widetilde{\mathbf{H}}\right)^{-1}\left(\boldsymbol{\Lambda}^{t-1}\right)^{-\frac{1}{2}}\mathbf{\Gamma}^{t-1}\\
			&=\mathbf{H}~~\underbrace{\left(\boldsymbol{\Lambda}^{t-1}\right)^{\frac{1}{2}}\left(\sigma^2\mathbf{I}_M+\left(\left(\boldsymbol{\Lambda}^{t-1}\right)^{\frac{1}{2}}\right)^H\mathbf{H}^H\mathbf{H}\left(\boldsymbol{\Lambda}^{t-1}\right)^{\frac{1}{2}}\right)^{-1}\left(\boldsymbol{\Lambda}^{t-1}\right)^{-\frac{1}{2}}\mathbf{\Gamma}^{t-1}}_{\mathbf{B}^{t-1}}		\end{align}
		
		From the definition of $\mathbf{B}^{t-1}$, we have
		\begin{align}\label{B's function}
			\left\{\alpha_k,u_k^{t-1},w_{k}^{t-1},\mathbf{h}_i\mathbf{h}_j^H, \forall i,j,k \right\}\rightarrow \mathbf{B}^{t-1}.
		\end{align}
		Moreover, $\v_k^t$ can be expressed as
		\begin{align}\label{eq: precoding vector linear expression}
			\mathbf{v}_k^t=\sum_{m=1}^Mb_{m,k}^{t-1}\mathbf{h}_m^H,
		\end{align}
		where $b_{m,k}^{t-1}$ is the element of row $m$ and column $k$ of $\mathbf{B}^{t-1}$.
		
		Substituting (\ref{eq: precoding vector linear expression}) into (\ref{eq: wmmse ut iteration}) and (\ref{eq: wmmse wt iteration}), we have
		\begin{align}\label{eq: u w's function}
		\left\{\left\{b_{i,j}^{t-1}\right\},\left\{\mathbf{h}_{i}\mathbf{h}_{j}^H\right\}\right\}\rightarrow \left\{u_k^{t},w_k^{t}\right\}.
		\end{align}
		Combining (\ref{B's function}) and (\ref{eq: u w's function}), we further have
		\begin{align}
			\left\{\alpha_k,u_k^{t-1},w_{k}^{t-1},\mathbf{h}_i\mathbf{h}_j^H\right\} \rightarrow \left\{u_k^{t},w_k^{t}\right\}.
		\end{align}
		Therefore, with a fixed initial variable, the final $\left\{u_k,w_k\right\}$ is only determined by $\left\{\alpha_k,\mathbf{h}_i\mathbf{h}_j^H,\forall i, j ,k\right\}$, which ends the proof.
	\end{proof}
	\subsection{Proof of Proposition \ref{prop: multi RB optimal solution structure}}
	The WSR maximization problem  $\mathscr{P}_2$ in MISO scenario is given as follows.
	\begin{align}
		\mathscr{P}_3:\max_{\left\{\mathbf{v}_k\right\}}&\sum_{k=1}^M\sum_{b=1}^B \alpha_{k,b}{\rm logdet}\left(1+\mathbf{h}_{k,b}\mathbf{v}_k\mathbf{v}_k^H\mathbf{h}_{k,b}^H\left(\sum_{m\neq k}\mathbf{h}_{k,b}\mathbf{v}_m\mathbf{v}_m^H\mathbf{h}_{k,b}^H+\sigma^2\right)^{-1}\right),\\
\mathbf{s.t.} &\sum_{k=1}^M{\rm Tr}\left(\mathbf{v}_k\mathbf{v}_k^H\right)\leq P_T.
	\end{align}
	Then the weighted least mean square error problem can be formulated as
	\begin{align}
		\mathscr{P}_4:\min_{\left\{u_{k,b},w_{k,b},\mathbf{v}_k\right\}} \sum_{k=1}^K\sum_{b=1}^B\alpha_{k,b}\left(w_{k,b}e_{k,b}-{\rm log}\left(w_{k,b}\right)\right)
	\end{align}
	where $e_{k,b}$ is given by
	\begin{align}
		e_{k,b}=\left(1-u_{k,b}^H\mathbf{h}_{k,b}\mathbf{v}_k\right)\left(1-u_{k,b}^H\mathbf{h}_{k,b}\mathbf{v}_k\right)^H+\sum_{m\neq k}u_{k,b}u_{k,b}^H\mathbf{h}_{k,b}\mathbf{v}_m\mathbf{v}_m^H\mathbf{h}_{k,b}^H \notag\\+\frac{\sum_{n=1}^{M}{\rm Tr}\left(\mathbf{v}_n\mathbf{v}_n^H\right)}{P_T}\sigma^2u_{k,b}^Hu_{k,b}
	\end{align}
	Similar to \cite{shi2011iteratively} and \cite{hu2020iterative}, it can be proven that $\mathscr{P}_3$ and $\mathscr{P}_4$ is equivalent.
	Therefore,  we can obtain the iterative based solution for $\mathscr{P}_3$ by solving $\mathscr{P}_4$ as follows.
	\begin{align}
		u_{k,b}&=\left(\sum_{m=1}^M\frac{\sigma^2}{P_T}\mathbf{v}_k^H\mathbf{v}_k+\mathbf{h}_{k,b}\mathbf{v}_m\mathbf{v}_m^H\mathbf{h}_{k,b}^H\right)^{-1}\mathbf{h}_{k,b}\mathbf{v}_{k}\\
		w_{k,b}&=\left(1-u_{k,b}^H\mathbf{h}_{k,b}\mathbf{v}_{k}\right)^{-1}\\
		\mathbf{v}_k&=\left(\sum_{m=1}^M\sum_{b=1}^B\alpha_{m,b}u_{m,b}w_{m,b}u_{m,b}^H\left(\frac{\sigma^2}{P_T}\mathbf{I}_{N_t}+\mathbf{h}_{m,b}^H\mathbf{h}_{m,b}\right)\right)^{-1}\left(\sum_{b=1}^B\mathbf{h}_{k,b}^H\alpha_{k,b}u_{k,b}w_{k,b}\right)\label{eq:multi RB v iterative formula}
	\end{align}	
	Therefore, the optimal precoding matrix $\mathbf{\v}_k^*$ should meet the form of (\ref{eq:multi RB v iterative formula}).
	Proposition \ref{prop: multi RB optimal solution structure} can be obtained by reorganizing (\ref{eq:multi RB v iterative formula}), which ends the proof.
	\newpage
	%	The conclusion of Proposition \ref{prop: multi RB optimal solution structure} can be obtained by sorting the iterative formula of $\v_{k}$ in (\ref{eq:multi RB v iterative formula}), which ends the proof.
	\bibliographystyle{ieeetr}
	\vspace{-6mm}
	\bibliography{BibDesk_File_v2}
	\vspace{-5mm}

	\end{document}

%% file: header.tex
\newtheorem{proposition}{Proposition}
\newtheorem{remark}{Remark}
\def\phi{\varphi}

\def\({\left(}
\def\){\right)}

\def\b0{{\mathbf{0}}}

\newcommand{\diag}{\mathrm{diag}}